\documentclass[11pt]{article}
\usepackage[margin=1in]{geometry}
\usepackage{amsmath,amssymb,amsthm,mathtools,booktabs}
\usepackage[T1]{fontenc}
\usepackage{lmodern}
\usepackage{microtype}
\usepackage{xcolor}
\usepackage{tikz}
\usepackage{float,enumitem}
\floatstyle{ruled}
\newfloat{algorithm}{tbp}{loa}
\floatname{algorithm}{Algorithm}
\usepackage{fancyhdr}
\usepackage[colorlinks=true,linkcolor=blue,urlcolor=blue,citecolor=blue]{hyperref}
\hypersetup{pdftitle={Coloring 3-colorable graphs with O(n^(4/23)) colors via a Gaussian-cover recursion},pdfauthor={Emile Anand},pdfsubject={Approximate graph coloring}}
\newtheorem{theorem}{Theorem}
\newtheorem{lemma}[theorem]{Lemma}
\newtheorem{proposition}[theorem]{Proposition}
\theoremstyle{definition}
\newtheorem{definition}[theorem]{Definition}
\theoremstyle{remark}
\newtheorem{remark}[theorem]{Remark}
\newcommand{\E}{\mathbb E}
\newcommand{\R}{\mathbb R}
\newcommand{\Var}{\operatorname{Var}}
\newcommand{\Prb}{\Pr}
\newcommand{\ip}[2]{\langle #1,#2\rangle}
\newcommand{\phib}{\overline\Phi}
\newcommand{\eps}{\varepsilon}

\title{Coloring 3-colorable graphs with $O(n^{4/23})$ colors\protect\\
via a Gaussian-cover recursion}
\author{Emile Anand\footnote{Georgia Institute of Technology, School of Computer Science, \texttt{emile@gatech.edu}}}
\date{}

\begin{document}
\maketitle

\begin{abstract}
We give a randomized polynomial-time algorithm that colors any promised
$3$-colorable graph on $n$ vertices with $\smash{O(n^{4/23}) =
O(n^{0.17391\ldots})}$ colors, improving on the recent bounds of $O(n^{0.19539})$
by Bansal, Huang, and Lee~\cite{BHL} and Narang and Tang \cite{narang2026improvedsdpcoloring3colorable} who obtained $\smash{O(n^{(13-\sqrt{97})/18+\epsilon})=O(n^{0.17506\dots + \epsilon})}$ colors for every fixed $\smash{\epsilon>0}$.

To prove our result, we start from a fixed-level semidefinite relaxation, where we use  a finite-depth recursion on Gaussian covers. Fixing a root vertex, we group vertices by correlation with the root vector. Here, each step extends a cover of directions by one edge and transfers it to a successor group. Our key analytic ingredient is a variance bound for Gaussian maxima: for a maximum of $m\geq 2$ centered linear forms with coefficient norms at most $r$, mean $\mu$, and variance $v$, we prove $v\leq r^2-\mu^2/(2\log m)$ using Chen's Gaussian convexity theorem \cite{Chen}. Together with a variance-scale lower-tail estimate, this controls the threshold loss at each extension, which shows that root-conditioned vector colorings can either extract a large independent set from a group or bound its size, forcing a contradiction after constantly many steps. The resulting sparse-case guarantee combines with the dense progress bound of Kawarabayashi, Thorup, and Yoneda~\cite{KTY}, and the recursion's numerical inequalities are verified via rational interval arithmetic.\looseness=-1
\end{abstract}

\vspace{0.5cm}

\small \noindent \textbf{AI disclosure.} This paper was written with the assistance of AI tools, including Cognition's Devin agent (using SWE-2 High), which identified the connection to Ehrhard's inequality. Through an extended conversation with OpenAI's Astra, we found that this approach yielded an improvement on this problem, and we used Devin to formalize the argument. The author has reviewed all proofs and arguments for correctness and takes full responsibility for the content.

\newpage

\section{Introduction}

Given an $n$-vertex graph promised to be 3-colorable, how many colors
suffice in polynomial time? The promise guarantees a partition into
three independent sets without providing it. Deciding 3-colorability
is NP-complete, and the promise does not remove the difficulty:
coloring a 3-colorable graph with five colors is already
NP-hard~\cite{KLS,GK,BKO}, and hardness for every constant number of
colors is conjectured. The problem is a canonical promise problem: a
global property of the input is guaranteed but not given, so
algorithms must exploit structure they cannot directly inspect. It has
been a test case for semidefinite rounding and convex hierarchies, and
the gap it leaves openn, which is constant-factor hardness against
polynomial-color algorithms, is among the largest in approximation
theory. Approximation algorithms therefore seek a proper
coloring with a controlled number of additional colors. Historically, progress on
this problem has relied on combining combinatorial structure with
semidefinite programming.

The two approaches apply in complementary regimes: combinatorial
methods exploit a large minimum degree, whereas semidefinite rounding
benefits from an upper bound on the maximum degree. In the latter regime, the basic task is to extract an
independent set large enough to make progress toward a coloring. The
Gaussian rounding algorithm of Karger, Motwani, and Sudan~\cite{KMS}
provides the starting point: they show how to convert a vector coloring into a random vertex set by thresholding Gaussian inner products, removing conflicts via a maximal matching. Notably, the geometry left behind when this
procedure has small expected output provides additional
information, since, at each selected vertex, the matching gives a family of edge directions whose Gaussian maximum remains large.

In this work, we develop a finite-depth recursion for exploiting that
geometry. We fix a root and group vertices into correlation shells
according to their inner product with its vector, as in the neighborhood analyses of~\cite{KMS,BHL}. A family in one such
shell is described by its correlation and a Gaussian threshold that it
reaches with inverse-polynomial probability. The recursion extends its
endpoints by one graph edge, projects away the root direction, and
transfers the cover property to a successor shell. Through this process, large shells yield
independent sets through rounding, and small shells impose cardinality
bounds that constrain the loss in the next extension. These alternatives produce a scalar recursion whose continuing states lie in
a compact domain (see Figure~\ref{fig:recursion}).

The cover-composition arguments of Bansal, Huang, and
Lee~\cite[Lemma~4.14]{BHL} and
Chlamt\'a\v{c}~\cite[Theorem~3.4.4]{Chl09} use
Borell's Gaussian isoperimetric enlargement inequality~\cite{Borell}.
In those bounds, the enlargement cost is multiplied by the radius of
the inner family. We instead derive a composition estimate from the
variance of the inner Gaussian maximum. For a maximum of $m\geq 2$
centered Gaussian linear forms with coefficient norms at most $r$,
mean $\mu$, and variance $v$, we prove
\[
  v\leq r^2-\frac{\mu^2}{2\log m}.
\]
Together with Chen's lower-deviation inequality~\cite{Chen}, this
replaces the full-radius factor by a variance-dependent factor.
The resulting composition inequality gives the neighborhood lift in
Lemma~\ref{lem:lift}. Section~\ref{sec:related} compares this proof with the concurrent approach
of Narang and Tang~\cite{narang2026improvedsdpcoloring3colorable}.
The remaining finite inequalities are verified by exact rational
interval arithmetic over an adaptive subdivision of the compact domain.\looseness=-1

Our main theorem is an improved coloring bound of  $O(n^{4/23})=O(n^{0.17391\ldots})$, as stated
in Theorem~\ref{thm:main}, which improves on the bound of Bansal, Huang, and
Lee~\cite{BHL} and concurrent work of Narang and Tang~\cite{narang2026improvedsdpcoloring3colorable}, where the exponent we present has a slightly smaller value by setting
$\alpha=1/(5+3c_*)$, where $c_*=1/4+10^{-6}$, to absorb the subpolynomial
losses in the dense--sparse reduction. For our analysis, only a fixed number of neighborhood extensions is needed, since the algorithm rounds the full root-correlation shells and does not have to enumerate the paths or construct the pruned graph used in the proof.

\begin{theorem}\label{thm:main}
Every promised 3-colorable graph on $n$ vertices has a proper coloring with
$O(n^{4/23})$ colors that can be found in randomized polynomial time, with
success probability at least $1-n^{-10}$.
\end{theorem}

\subsection{Related work}\label{sec:related}

\textbf{Combinatorial and semidefinite algorithms.}
Progress on this problem has alternated between combinatorial and
semidefinite tools. Wigderson~\cite{Wigderson} gave an
$O(n^{1/2})$-coloring algorithm by using the bipartite neighborhoods
of high-degree vertices and greedily coloring the remaining graph.
Blum~\cite{Blum} introduced the general notion of progress toward a
coloring and improved the guarantee to $\widetilde O(n^{3/8})$ by
looking one step further out, exploiting distance-2 neighborhoods
where earlier arguments had used only immediate ones. Karger, Motwani,
and Sudan~\cite{KMS} brought in semidefinite programming, following
Goemans and Williamson~\cite{GW}, to reach $\widetilde O(n^{1/4})$
colors; Blum and Karger~\cite{BK} then combined the two approaches
for $\widetilde O(n^{3/14})$. Arora, Chlamt\'a\v c, and
Charikar~\cite{ACC} brought the bound to $O(n^{0.2111})$ through a
geometric analysis of semidefinite rounding, and
Chlamt\'a\v c~\cite{Chl07} pushed it to $O(n^{0.2072})$ with a
hierarchy of semidefinite relaxations~\cite{Lasserre}. Kawarabayashi
and Thorup~\cite{KT17} returned to combinatorial methods, which,
combined with the semidefinite analysis, give $\widetilde O(n^{0.19996})$
colors. For general graphs without a colorability promise,
Halld\'orsson~\cite{Halldorsson} achieved an
$O(n(\log\log n)^2/\log^3 n)$ approximation ratio.

Kawarabayashi, Thorup, and Yoneda~\cite{KTY} subsequently developed a combinatorial
progress algorithm for graphs of minimum degree $\Delta>\sqrt n$ that
colors such graphs with $n^{o(1)}\sqrt{n/\Delta}$ colors; this idea, combined
with semidefinite rounding, gives their $\widetilde O(n^{0.19747})$
bound. We make use of their dense-side guarantee in our work. Related structural parameters for dense graphs
include rank-expansion, a well-linkedness measure defined through the
cutrank of balanced separations~\cite{Anand}.

\textbf{Hardness.}
Khanna, Linial, and Safra~\cite{KLS} showed that coloring a
3-colorable graph with four colors is NP-hard; Guruswami and
Khanna~\cite{GK} gave a new proof that does not rely on the PCP
theorem, and Bul\'in, Krokhin, and Opr\'sal~\cite{BKO} extended the
hardness to five colors. Dinur, Mossel, and Regev~\cite{DMR} studied
conditional hardness for approximate coloring under variants of the
Unique Games Conjecture. Indeed, without the 3-colorability promise, the
chromatic number is hard to approximate within $n^{1-\eps}$ for every
$\eps>0$~\cite{FK,Zuckerman}.

\textbf{Neighborhood covers and root-conditioned rounding.}
Bansal, Huang, and Lee~\cite{BHL} obtained an
$O(n^{0.19539})$-coloring algorithm by extending the neighborhood
analysis to a third step and constructing a vector $5/2$-coloring
for the resulting third-neighborhood sets. We adopt their
root-conditioned vector-coloring constructions and their
dense-sparse formulation. The cover machinery itself, which comprises of Gaussian
covers, packings, and the composition and pruning steps, was assembled across these earlier SDP analyses, and the statements we use
appear in Chlamt\'a\v{c}'s dissertation~\cite{Chl09}, whose Theorems~3.4.4 and~3.6.9 establish cover composition and the packing
spread theorem. Whereas the neighborhood analysis in~\cite{BHL}
uses three levels, our proof combines the norm-sensitive composition
of Lemma~\ref{lem:compose} with root-conditioned shell rounding and extends the
recursion to a fixed depth of $404$.

\textbf{Concurrent work.}
Concurrently, Narang and Tang~\cite{narang2026improvedsdpcoloring3colorable} obtain
$O(n^{(13-\sqrt{97})/18+\varepsilon})$ colors for every fixed
$\varepsilon>0$, also via a fixed-depth Gaussian-cover
recursion. Their Lemma~2.5 yields the same leading
composition loss as Lemma~\ref{lem:compose} here for a common inner norm bound,
using Ehrhard's inequality. An explicit potential then establishes
their guarantee for modified KMS rounding.
Our proof derives the composition estimate from Chen's Gaussian
convexity theorem and a variance bound, combines the recursion with
root-conditioned shell rounding from a three-round SoS relaxation,
and verifies its drift by exact rational interval arithmetic, to obtain the exponent $4/23$.

\textbf{Gaussian convexity and finite verification.}
Gaussian isoperimetry controls the measure of Euclidean enlargements of
sets~\cite{Borell}; Ehrhard's inequality describes Gaussian measure under
convex combinations of convex sets~\cite{Ehrhard}. The cover-composition
proof in Bansal, Huang, and Lee~\cite[Lemma 4.14]{BHL} uses the
isoperimetric enlargement inequality. Here the inner maximum is a convex
Gaussian functional, so its variance provides additional information
about the lower tail, so the resulting variance bound is an instance of
Gaussian  superconcentration~\cite{Chatterjee}. Appendix~\ref{app:entropy} gives a weighted composition statement
that retains individual witness masses and a precise entropy truncation
bound, and the recursion itself uses the cardinality-based estimate of
Lemma~\ref{lem:compose}. Chen~\cite{Chen} proves convexity of the normalized logarithmic
moment-generating function of a convex Gaussian functional, together
with a lower-tail estimate at the scale of its variance, and we apply both of these tools
to a maximum of linear forms. A finite exact computation verifies the
initialization and the uniform positive drift of the resulting
recursion, stated as Proposition~\ref{prop:finite}.

\paragraph{Notation.} 
For a graph $G=(V,E)$, write $n=|V|$, let $G[S]$ denote the
subgraph induced by $S\subseteq V$, and let $N_G(v)$ denote the
open neighborhood of $v$. We write $\Delta(G)$ for the maximum
degree and use $\Delta$ for a specified upper bound on $\Delta(G)$.
All vector norms are Euclidean, and $\ip{x}{y}$ denotes the
Euclidean inner product. Throughout, $g$ is a standard Gaussian
vector in the ambient space, and $\E$ and $\Var$ denote expectation
and variance. For $Z\sim N(0,1)$, define $\Phi(t):=\Prb[Z\le t]$ and $\phib(t):=1-\Phi(t)=\Prb[Z\ge t]$. We write $x_+:=\max\{x,0\}$, and all logarithms are natural.
The notation $\widetilde O$ suppresses polylogarithmic factors.
Unless stated otherwise, asymptotic notation refers to
$n\to\infty$, with the recursion depth and all auxiliary constants
held fixed. Moreover, all the neighborhoods we consider are open neighborhoods. Unless explicitly restricted,
a root-correlation shell with root $i$ consists of vertices in
$V\setminus\{i\}$ of the original input graph. All analytic and
algorithmic correlation buckets use one fixed partition of width
$1/\log\Delta$. \\

\section{Technical Background}\label{sec:architecture}

We recall the rounding algorithms and the geometric statements used
in the proof. Throughout this section, $g$ denotes a standard Gaussian
vector in the relevant Euclidean space.

\subsection{Vector colorings and Gaussian rounding}
\label{sec:background-vector}

\begin{definition}[Vector coloring {
\cite{KMS}}]\label{def:vector-coloring}
For a real number $\kappa\ge2$, a \emph{vector $\kappa$-coloring}
of $G=(V,E)$ assigns unit vectors $(x_j)_{j\in V}$ such that
$\ip{x_j}{x_k}\le-1/(\kappa-1)$ for every $jk\in E$.
The coloring is \emph{strict} if equality holds on every edge.
\end{definition}

A proper integer $\kappa$-coloring gives such vectors by placing
its colors at the vertices of a regular simplex. Karger, Motwani,
and Sudan \cite{KMS} round a vector coloring by selecting vertices above a
Gaussian threshold and keeping the isolated vertices.

\begin{algorithm}[H]
\caption{KMS rounding {\cite{KMS}}}\label{alg:kms}
\small
\textbf{Input:} a graph $H$, a vector coloring $(x_j)$, and a threshold $u\ge0$.
\begin{enumerate}[leftmargin=1.6em,itemsep=0pt,topsep=2pt]
  \item Sample $g\sim N(0,I)$
  \item Construct the set
        $T\gets\{j\in V(H):\ip{g}{x_j}\ge u\}$.
  \item \textbf{Return} $\{j\in T:N_H(j)\cap T=\varnothing\}$.
\end{enumerate}
\end{algorithm}

\begin{lemma}[Karger--Motwani--Sudan {
\cite{KMS}}]\label{lem:kms-background}
Given a vector $\kappa$-coloring of an $m$-vertex graph with maximum
degree at most $\Delta\ge2$, a suitable threshold in
Algorithm~\ref{alg:kms} gives an independent set $I$ with
$\E|I|\ge\widetilde\Omega(m\Delta^{-(\kappa-2)/\kappa})$.
\end{lemma}

We use the real-$\kappa$ formulation in
Bansal, Huang, and Lee~\cite[Theorem~2.3]{BHL}.
Arora, Chlamt\'a\v{c}, and Charikar~\cite{ACC} replace the isolated
vertex step by a maximal matching which retains every
selected isolated vertex (and hence Lemma~\ref{lem:kms-background} applies to it).
We will also run it at a lower threshold, where its matching gives
the geometric information needed for the recursion.

\begin{algorithm}[H]
\caption{Modified KMS rounding {\cite{ACC}}}\label{alg:threshold-matching}
\small
\textbf{Input:} a graph $H$, unit vectors $(x_j)$, and a threshold $u\ge0$.
\begin{enumerate}[leftmargin=1.6em,itemsep=0pt,topsep=2pt]
  \item Sample $g\sim N(0,I)$ and set
        $T\gets\{j\in V(H):\ip{g}{x_j}\ge u\}$.
  \item Find a maximal matching $M$ in $H[T]$, using a fixed edge ordering.
  \item \textbf{Return} $T\setminus V(M)$, where $V(M)$ denotes the matched vertices.
\end{enumerate}
\end{algorithm}

We obtain our vectors from the symmetrized three-round SoS/Lasserre
relaxation. The following properties explain why conditioning on a
root gives additional rounding options.

\begin{lemma}[Bansal--Huang--Lee {
\cite[Section~2.1]{BHL}}]\label{lem:sos-background}
An exact feasible symmetrized three-round solution gives a strict
vector 3-coloring $(v_j)$ and consistent local distributions on proper
colorings of triples, with uniform one-vertex marginals.
For a root $i$, write $q_j=\ip{v_i}{v_j}$. Then
$\Prb[j=\mathsf R\mid i=\mathsf R]=(1+2q_j)/3$,
so $q_j\ge-1/2$ and $q_j+q_k\le1/2$ whenever $jk\in E$.
\end{lemma}

A \emph{root-correlation shell} consists of vertices whose $q_j$
lie in one interval of the fixed partition. Lemma~\ref{lem:shell-rounding}
constructs and rounds vector colorings on these shells.
Section~\ref{sec:algorithm} treats the accuracy of the computed SoS solution.

\subsection{Gaussian covers and packings}
\label{sec:background-packing}

On an edge $ij$ of a strict vector 3-coloring, let
$v_{ij}=(2v_j+v_i)/\sqrt3$. This is a unit vector perpendicular to
$v_i$. If both endpoints pass threshold $t$, then
$\ip{g}{v_{ij}}\ge\sqrt3t$. If a matching rounding fails, this forces large Gaussian maxima among these tangent directions.

\begin{definition}[Gaussian cover {
\cite{ACC,Chl09}}]\label{def:gaussian-cover}
For $s\in\R$ and $\delta\in(0,1]$, a finite family
$X=(r_j)_{j\in J}$ is an \emph{$(s,\delta)$-cover} if
$\Prb[\max_{j\in J}\ip{g}{r_j}\ge s]\ge\delta$.
\end{definition}

\begin{lemma}[Arora--Chlamt\'a\v{c}--Charikar {
\cite[Lemmas~3--4]{ACC}}]\label{lem:rounding-cover-background}
Let $I$ be the output of Algorithm~\ref{alg:threshold-matching}
on an $n$-vertex graph with a strict vector 3-coloring and threshold
$t>0$. Either $\E|I|\ge n\phib(t)/4$, or some nonempty induced
subgraph $H$ has the property that
$(v_{ij})_{j\in N_H(i)}$ is a $(\sqrt3t,1/8)$-cover for every
$i\in V(H)$.
\end{lemma}

To restrict covers while controlling their probability, we keep
weights on the possible witnesses.

\begin{definition}[Packing and inefficiency {
\cite[Section~3.6]{Chl09}}]\label{def:gaussian-packing}
An \emph{$(s,\delta)$-packing} on $X=(r_j)_{j\in J}$ assigns
nonnegative weights $\mu(j)$ with $\mu(J)\ge\delta$ and
\[
 \mu(A):=\sum_{j\in A}\mu(j)
 \le\Prb\!\left[\max_{j\in A}\ip{g}{r_j}\ge s\right]
 \quad\text{for every }A\subseteq J.
\]
For unit directions, the packing is said to be \emph{at most $c$-inefficient}
if $|J|\le\phib(s)^{-(1+c)}$.
\end{definition}

Assigning each successful Gaussian outcome to its first witness
turns a cover into a packing. Then, restriction preserves the packing
inequality, so one class of any $B$-part partition retains mass
at least $\delta/B$. The following result gives a more geometric
restriction.

\begin{lemma}[Chlamt\'a\v{c} {
\cite[Theorem~3.6.9]{Chl09}}]\label{lem:spread-background}
Fix $c>0$. For all sufficiently large $s$, every at most
$c$-inefficient $(s,\delta)$-packing of unit vectors has a
restriction $J'\subseteq J$ such that $\mu(J')\ge\delta-1/\log s$
and for all $j\in J'$, $\mu\{k\in J':\ip{r_j}{r_k}\ge\lambda\}\le p$ and $\lambda=\frac{c}{1+c}\left(1+\frac{C_1}{\log s}\right)$, where $p=e^{-C_2(\log s)^2}$,
and where  $C_1,C_2>0$ only depend on $c$.
\end{lemma}

This is the form restated by Bansal, Huang, and Lee
\cite[Theorem~4.11]{BHL}. For matching rounding, the symmetric
weights $\mu_i(j)=\Prb[ij\in M]/\phib(t)$ form outgoing packings
at threshold $\sqrt3t$ (Appendix~\ref{app:approx}).
We apply Lemma~\ref{lem:spread-background} with a fixed
$c_{\rm pr}>c_*$ and prune to a graph whose outgoing covers and
two-walk supports persist under extension. We provide this buffered form in Appendix~\ref{app:pruning}. Importantly, the algorithm rounds full shells and does not construct the packings.\\

\subsection{Gaussian convexity and lower tails}
\label{sec:background-chen}

We use the following specialization of Chen's results to control
the loss when we compose two covers.

\begin{lemma}[Chen {
\cite[Theorems~1 and~3]{Chen}}]\label{lem:chen-background}
Let $F:\R^N\to\R$ be convex and Lipschitz, and put
$\mu=\E F(g)$ and $v=\Var(F(g))$.
The function $\Lambda_F(\lambda)=\lambda^{-1}\log\E e^{\lambda F(g)}$,
extended by $\Lambda_F(0)=\mu$, is convex on $\R$.
Moreover, if $v>0$, then
$\Prb[F(g)\le\mu-x]\le e^{-x^2/(2v)}$ for every $x>0$.
If $v=0$, then $F(g)=\mu$ almost surely.
\end{lemma}

A finite maximum of linear forms is convex and Lipschitz.
Lemma~\ref{lem:variance} applies Chen's convexity to bound its
variance in terms of its mean, cardinality, and vector norms.
The lower-tail estimate then gives the composition bound in
Lemma~\ref{lem:compose}.

\subsection{Dense--sparse progress and parameters}
\label{sec:background-progress}

\begin{definition}[Progress toward a coloring {
\cite{Blum}}]\label{def:coloring-progress}
For a promised 3-colorable $n$-vertex graph, \emph{progress toward
an $f(n)$-coloring} means finding one of the following:
\begin{enumerate}[label=(\roman*),leftmargin=2em,itemsep=0pt,topsep=2pt]
  \item an independent set of size $\Omega(n/f(n))$;
  \item a nonempty independent set $I$ with $|N_G(I)|=O(f(n)|I|)$;
  \item distinct vertices with the same color in every proper 3-coloring.\\
\end{enumerate}
\end{definition}

Our sparse routine provides the first type. We use the following
two results for dense graphs and for combining the two degree regimes.

\begin{lemma}[Kawarabayashi--Thorup--Yoneda {
\cite[Theorem~1]{KTY}}]\label{lem:dense-background}
On a promised 3-colorable $n$-vertex graph of minimum degree at
least $\Delta>\sqrt n$, a polynomial-time algorithm makes progress
toward an $n^{o(1)}\sqrt{n/\Delta}$-coloring.
\end{lemma}

\begin{lemma}[Dense--sparse combination {
\cite[Theorem~2.6]{BHL}}]\label{lem:combination-background}
Fix $a,b\in(0,1)$. Suppose polynomial-time algorithms make progress
toward an $\widetilde O(n^a)$-coloring whenever the maximum degree
is at most $n^b$, and whenever the minimum degree is at least $n^b$.
Then a polynomial-time algorithm colors every promised 3-colorable
$n$-vertex graph with $\widetilde O(n^a)$ colors.
\end{lemma}

We use the formulation of Bansal, Huang, and Lee above, based on
Kawarabayashi and Thorup~\cite{KT17}. Balancing our sparse cost
$\Delta^{1/D}$ with the dense cost $\sqrt{n/\Delta}$ gives
$\Delta=n^{D/(D+2)}$. We fix
\begin{equation}\label{eq:parameters}
\begin{aligned}
 c_*&:=\frac14+\frac1{10^6},
 &D&:=3(1+c_*),
 &d&:=\frac{c_*}{1+c_*},
&\alpha&:=\frac1{5+3c_*},
 &\beta&:=\frac{D}{5+3c_*}.
\end{aligned}
\end{equation}
Set $\Delta=n^\beta$ and choose $t$ by
$\phib(t)=\Delta^{-1/D}$, so the sparse target is
$n\phib(t)=n^{1-\alpha}$. For the shell bounds, put
\begin{equation}\label{eq:scales}
\begin{aligned}
 \log\Delta&=\frac D2t^2+o(t^2),
 &R_*&:=\frac{4+3c_*}{D}, &S&:=R_*+\frac1{10000},
 &\Delta^{R_*}&=n^{1-\alpha}.
\end{aligned}
\end{equation}
The slack $S-R_*$ absorbs shell-rounding losses, and
$\alpha<4/23$ absorbs the subpolynomial costs of the coloring
reduction. Thus, negligible terms are uniform for recursion depth at most $404$.
\section{Gaussian maxima and composition}\label{sec:gaussian}
A high cover threshold forces the inner Gaussian maximum to have a
large mean. The following bound then limits its variance and, hence,
the loss in composing covers.

\begin{lemma}\label{lem:variance}
Let $F(g)=\max_{1\le j\le m}\ip{g}{u_j}$, where $g$ is standard Gaussian,
$m\ge2$, and $\|u_j\|\le r$. Write $\mu=\E F$ and $v=\Var(F)$.
Then
\begin{equation}\label{eq:variance}
 v\le r^2-\frac{\mu^2}{2\log m}.
\end{equation}
For $x\ge0$, $\Prb(F\le\mu-x)\le\exp(-x^2/(2v))$, with the usual
interpretation when $v=0$.
\end{lemma}
\begin{proof}
The function $F$ is convex and $r$-Lipschitz. Chen's Theorems~1 and~3
\cite{Chen} give convexity of
$\Lambda(\lambda)=\lambda^{-1}\log\E e^{\lambda F}$,
continuously extended at zero, and the stated variance-scale lower tail.
The expansion
$\log\E e^{\lambda F}=\lambda\mu+\lambda^2v/2+o(\lambda^2)$
as $\lambda\to0$ gives $\Lambda(0)=\mu$ and
$\Lambda'(0)=v/2$, so convexity implies
$\Lambda(\lambda)\ge\mu+\lambda v/2$ for $\lambda>0$.
For such $\lambda$, we have
\begin{align*}
 \E e^{\lambda F}
   &=\E\max_{1\le j\le m}e^{\lambda\ip{g}{u_j}}\\
   &\le\sum_{j=1}^m\E e^{\lambda\ip{g}{u_j}}\\
   &=\sum_{j=1}^m\exp\!\left(\frac{\lambda^2\|u_j\|^2}{2}\right)\le m\exp\!\left(\frac{\lambda^2r^2}{2}\right),
    \end{align*}
    and $\mu+\frac{\lambda v}{2}
   \le\Lambda(\lambda)
    \le\frac{\log m}{\lambda}+\frac{\lambda r^2}{2}$ with $v\le r^2-\frac{2\mu}{\lambda}
                  +\frac{2\log m}{\lambda^2}$.
Since $\E F\ge\E\ip{g}{u_1}=0$, we have $\mu\ge0$.
For $\mu>0$, substituting $\lambda=2\log(m)/\mu$ gives
\begin{align*}v &\le r^2-\mu^2/\log m+\mu^2/(2\log m)\\
&=r^2-\mu^2/(2\log m).\end{align*} For $\mu=0$, taking $\lambda\to\infty$
gives $v\le r^2$, and both cases give \eqref{eq:variance}.
\end{proof}

Now, recall that an $(s,\delta)$-cover has Gaussian maximum at least $s$
with probability at least $\delta$. Consider an inner cover of at most
$\Delta$ vectors of norm at most $r\le1$, with threshold
$s_t=\sqrt3t-o(t)$ and probability at least $\delta_t=t^{-O(1)}$.
For its maximum $F$ and mean $\mu=\E F$, Gaussian concentration gives
\begin{align*}
 \delta_t
   &\le\Prb(F\ge s_t)
    \le\exp\!\left(-\frac{(s_t-\mu)_+^2}{2r^2}\right),
\end{align*}
with \begin{align*}\mu&\ge s_t-r\sqrt{2\log(1/\delta_t)}\\
    &\ge\sqrt3t-o(t)-O(\sqrt{\log t})
     \\
     &=\sqrt3t-o(t).
\end{align*}
Here $r>0$ for large $t$, since a zero family cannot reach $s_t>0$.
A single distinct vector would give
$\Prb(F\ge s_t)\le e^{-s_t^2/(2r^2)}
\le e^{-(3/2-o(1))t^2}=o(\delta_t)$, also impossible.
Thus $2\le m\le\Delta$, and Lemma~\ref{lem:variance} gives
\begin{align}
 \Var(F)
   &\le r^2-\frac{\mu^2}{2\log m}
    \\
    &\le r^2-\frac{\mu^2}{2\log\Delta}\notag\\
   &\le r^2-\frac{3t^2-o(t^2)}{Dt^2+o(t^2)}\notag\\
   &=r^2-\frac1{1+c_*}+o(1).\label{eq:normvariance}
\end{align}
Nonnegativity of the variance therefore requires
$r^2\ge1/(1+c_*)-o(1)$, and all the errors are uniform whenever the threshold
and cover-probability assumptions are uniform.

\begin{lemma}[Norm-sensitive composition]\label{lem:compose}
Fix constants $L>0$ and $A\ge0$. Suppose a family of unit vectors $(x_j)_{j\in J}$ is an
$(Lt,t^{-O(1)})$-cover, with $|J|\le\Delta^{A+o(1)}$. For each $j$, let
$U_j\subset x_j^\perp$ be a fixed inner $(\sqrt3t-o(t),t^{-O(1)})$-cover,
with at most $\Delta$ vectors and a uniform norm upper bound $r_j\le1$.
Then there is a joint witness, on an event of probability $1/\mathrm{poly}(t)$, satisfying $\ip{g}{x_j}\ge Lt$ and
\begin{align*}
 \max_{u\in U_j}\ip{g}{u}
 &\ge\left(\sqrt3-
 \sqrt{\left(r_j^2-\frac1{1+c_*}\right)_+}
 \sqrt{(DA-L^2)_+}-o(1)\right)t.
\end{align*}
The error is uniform when $L$ ranges over a fixed compact subset of
$(0,\infty)$ and $A$ ranges over a fixed bounded interval, provided
the cover-probability and inner-threshold bounds are uniform.
\end{lemma}
\begin{proof}
Let $F_j(g)=\max_{u\in U_j}\ip{g}{u}$,
$\mu_j=\E F_j$, and $v_j=\Var(F_j)$. Set
$H_j=\{\ip{g}{x_j}\ge Lt\}$, $H=\bigcup_{j\in J}H_j$,
$\delta=\Prb(H)$, and
$B_t=\log(2|J|\phib(Lt)/\delta)$. Since the $x_j$ are unit vectors, we have that $ B_t\ge\log2$, and 
\begin{align*}
 t^{-O(1)}&\le\delta
       \le\sum_{j\in J}\Prb(H_j)=|J|\phib(Lt).
\end{align*}
Let $E_j=\{F_j<\mu_j-\sqrt{2v_jB_t}\}$.
Lemma~\ref{lem:variance} gives $\Prb(E_j)\le e^{-B_t}$ when
$v_j>0$; when $v_j=0$, the strict inequality makes $E_j$ null.
Since $U_j\subset x_j^\perp$, the event $E_j$ is independent of $H_j$.
Therefore
\begin{align*}
 \Prb\!\left(\bigcup_{j\in J}(H_j\cap E_j)\right)
   &\le\sum_{j\in J}\Prb(H_j\cap E_j)\\
   &=\sum_{j\in J}\Prb(H_j)\Prb(E_j)\\
   &\le |J|\phib(Lt)e^{-B_t}=\frac\delta2,
   \end{align*}
   and
   \begin{align*}
 \Prb\!\left(\bigcup_{j\in J}(H_j\cap E_j^c)\right)
   &\ge\Prb(H)-\Prb\!\left(\bigcup_{j\in J}(H_j\cap E_j)\right)\ge\frac\delta2.
\end{align*}
With probability at least $\delta/2$, some $j$ satisfies both
$\ip{g}{x_j}\ge Lt$ and $F_j\ge\mu_j-\sqrt{2v_jB_t}$.
To bound this threshold, we use Gaussian-tail asymptotics and $|J|$ to get
\begin{align*}
 B_t
   &=\log|J|+\log\phib(Lt)+\log(2/\delta)\\
   &\le(A+o(1))\log\Delta-\frac{L^2t^2}{2}+O(\log t)\\
   &=\frac{DA-L^2}{2}t^2+o(t^2)
    \\
    &\le\frac{(DA-L^2)_+}{2}t^2+o(t^2).
\end{align*}
Put $\sigma_j^2=(r_j^2-1/(1+c_*))_+$ and $Q=(DA-L^2)_+$.
The mean estimate and \eqref{eq:normvariance} allow a uniform
$\eta_t\ge0$, tending to zero, such that
$\mu_j\ge(\sqrt3-\eta_t)t,
 v_j\le\sigma_j^2+\eta_t$, and
$2B_t\le(Q+\eta_t)t^2$, and \begin{align*}
 \sqrt{2v_jB_t}
   &\le t\sqrt{(\sigma_j^2+\eta_t)(Q+\eta_t)}\\
   &\le t(\sigma_j+\sqrt{\eta_t})(\sqrt Q+\sqrt{\eta_t})\\
   &=t\bigl(\sigma_j\sqrt Q+o(1)\bigr).\end{align*} Thus, we have that
$\mu_j-\sqrt{2v_jB_t}\ge\bigl(\sqrt3-\sigma_j\sqrt Q-o(1)\bigr)t$. The error is uniform because $\sigma_j$ and $Q$ are bounded, including
when either is zero. So, desired threshold holds with
probability at least $\delta/2=t^{-O(1)}$.
\end{proof}

For a common norm bound $r$, setting $E=DA$, $\lambda=L$, $U=r$,
and $c=c_*$ in~\cite[Lemma~2.5]{narang2026improvedsdpcoloring3colorable}
gives the same leading loss
$\sqrt{(DA-L^2)(r^2-1/(1+c_*))}$ on the feasible range.
The formulations differ in their cover probabilities and norm bounds.
Appendix~\ref{app:entropy} gives a weighted form of our estimate.\\

\paragraph{Bucketing and projection conventions.}
We fix the families before the Gaussian in Lemma~\ref{lem:compose} is
sampled. Order a cover $(a_j)_{j\in J}$, and let $W_j$ be the event
that $j$ is the first index with $\ip{g}{a_j}\ge s$. For
$w_j=\Prb(W_j)$ and a partition $J=J_1\sqcup\cdots\sqcup J_B$,
disjointness gives
\begin{align*}
 \sum_{j\in J}w_j
   &=\Prb\!\left(\max_{j\in J}\ip{g}{a_j}\ge s\right)
     \ge\delta,
     \end{align*} Therefore, we have that
$\sum_{b=1}^B\sum_{j\in J_b}w_j\ge\delta$ and $ \max_{1\le b\le B}\sum_{j\in J_b}w_j\ge\frac\delta B$, and so \[\Prb\!\left(\max_{j\in J_b}\ip{g}{a_j}\ge s\right)\ge\Prb\!\left(\bigcup_{j\in J_b}W_j\right)
    =\sum_{j\in J_b}w_j.\]
A class of mass at least $\delta/B$ is an $(s,\delta/B)$-cover. Now, for each outer index, select a heavy inner correlation bucket, and then
retain a heavy class of outer indices with the same chosen bucket. We may restrict the root correlation in the same way if needed. A fixed number
$k$ of restrictions into $B_\ell=t^{O(1)}$ classes leaves probability
$\delta_{\rm retained}\ge\delta_{\rm original}
\prod_{\ell=1}^kB_\ell^{-1}=t^{-O(1)}$.
The norm bound in \eqref{eq:normvariance} thus applies to the entire
selected inner family.

For projection, fix a unit vector $e$ and a family $U$ with
$|\ip{u}{e}|\le1$. Writing $P_{e^\perp}u=u-\ip{u}{e}e$, we have $ \ip{g}{P_{e^\perp}u}-\ip{g}{u}=-\ip{u}{e}\ip{g}{e}$ and $\sup_{u\in U}\bigl|\ip{g}{P_{e^\perp}u}-\ip{g}{u}\bigr|\le|\ip{g}{e}|$, with\[\Prb\!\left(\sup_{u\in U}
   \bigl|\ip{g}{P_{e^\perp}u}-\ip{g}{u}\bigr|>\log t\right)
   \le2e^{-(\log t)^2/2}.\]
So, an $(s,\delta)$-cover with $\delta\ge t^{-C}$ projects to a
cover with threshold $s-\log t$ and probability bounded below by
\begin{align*}
 \Prb\!\left(\max_{u\in U}\ip{g}{P_{e^\perp}u}
                   \ge s-\log t\right)
   &\ge\delta-2e^{-(\log t)^2/2}\\
   &\ge\delta\left(1-2e^{-(\log t)^2/2+C\log t}\right)
    \\
    &=(1-o(1))\delta.
\end{align*}
We apply this estimate separately to each inner family before choosing
its heavy bucket; no simultaneous projection event over all parents is
needed.

With bucket width $h=1/\log\Delta$, replacing correlations by bucket
representatives changes the coefficients by at most $\omega(h)=o(1)$,
uniformly on the compact parameter ranges below. For any family of
unit vectors $(a_\ell)_{\ell=1}^N$ with $N\le n^{C_0}$,
\begin{align*}
 \Prb\!\left(\max_{1\le\ell\le N}|\ip{g}{a_\ell}|
              >\sqrt{2(C_0+2)\log n}\right)
   &\le\sum_{\ell=1}^N
        2e^{-(C_0+2)\log n}\\
   &\le2n^{-2}.
\end{align*}
Outside this event, the error is at most
$\omega(h)\max_\ell|\ip{g}{a_\ell}|
\le\omega(h)\sqrt{2(C_0+2)\log n}=o(t)$,
since \eqref{eq:scales} and $\Delta=n^\beta$ give
$\log n=\Theta(t^2)$. At fixed depth there are only polynomially many
graph-indexed vectors, and $n^{-2}=o(t^{-C})$ for every fixed $C$.
Thus bucket substitutions have uniform $o(t)$ threshold errors.

\subsection{Graph input and shell rounding}\label{sec:shells}
The symmetrized three-round SoS relaxation gives unit vectors $v_i$
with $\ip{v_i}{v_j}=-1/2$ on edges and consistent triple distributions.
For an edge $ij$, the unit tangent vector
$v_{ij}=(2v_j+v_i)/\sqrt3$ satisfies $v_{ij}\perp v_i$ and
$v_j=-v_i/2+\sqrt3v_{ij}/2$.
Fix a root $i$ and write $q_j=\ip{v_i}{v_j}$.
Color symmetry gives
$p_j:=\Prb[j=\mathsf R\mid i=\mathsf R]=(1+2q_j)/3$;
since $0\le p_j\le1$, we have $-1/2\le q_j\le1$.
For an edge $jk$, the endpoints cannot both have the root's color, so $p_j+p_k
 =\Prb[j=\mathsf R\text{ or }k=\mathsf R\mid i=\mathsf R]\le1$ and $\frac{2+2q_j+2q_k}{3}\le1$. Thus, writing $q=q_j$ and $z=q_k$, we obtain
$q+z\le\tfrac12$.

The shell-rounding loss will be bounded by
\begin{equation}\label{eq:rho}
\rho(q)=\begin{cases}
 \dfrac{1+2q}{3-6q},&-1/2\le q\le0,\\[1mm]
 \dfrac{1-4q^2}{3},&0<q<1/16,\\[1mm]
 \min\left\{\dfrac{1-4q^2}{3},\dfrac{1-4q}{2+4q}\right\},
     &1/16\le q<1/4,\\[1mm]
 0,&q\ge1/4.
\end{cases}
\end{equation}
\begin{lemma}[Rounding correlation shells]\label{lem:shell-rounding}
Fix a root $i$ and $m$ vertices $H\subseteq V\setminus\{i\}$ with
$|\ip{v_i}{v_j}-q|\le h$, where $q\in[-1/2,1]$ and $h=o(1)$.
From the symmetrized three-round SoS solution one can compute, in
polynomial time, an independent set in $G[H]$ of expected size
\[
 m\Delta^{-\rho(q)-o(1)},
\]
with $\rho$ given in \eqref{eq:rho}. The bound is uniform in $q$,
including under the inverse-polynomial SDP errors described in
Section~\ref{sec:algorithm}.
\end{lemma}
\begin{proof}
Let $v_{\{(i,\mathsf R),(j,a)\}}$ be the SoS vector for the indicated
partial color assignment, with $a\in\{\mathsf R,\mathsf G,\mathsf B\}$.
Since $\Prb[i=\mathsf R]=1/3$, the conditioned vectors
$w_{j,a}:=\sqrt3\,v_{\{(i,\mathsf R),(j,a)\}}$ satisfy $\|w_{j,a}\|^2=\Prb[j=a\mid i=\mathsf R]$ with 
$\ip{w_{j,a}}{w_{k,b}}=\Prb[j=a,\ k=b\mid i=\mathsf R]$ and $\sum_a\|w_{j,a}\|^2=1$, where $\ip{w_{j,a}}{w_{j,b}}=0$ for $a\neq b$.

For an edge $jk$, the same-color inner products vanish. In particular,
with $p_j=\|w_{j,\mathsf R}\|^2=(1+2q_j)/3$, we have $\ip{w_{j,\mathsf R}}{w_{k,\mathsf R}}=0$, with $\Prb[j\ne\mathsf R,\ k\ne\mathsf R\mid i=\mathsf R]=1-p_j-p_k$ and $\ip{w_{j,\mathsf G}}{w_{k,\mathsf B}}
 +\ip{w_{j,\mathsf B}}{w_{k,\mathsf G}}
=1-p_j-p_k$.

For $q\le0$, we use the construction of \cite[Lemma~3.2]{BHL}: let $x_j\coloneqq \frac{w_{j,\mathsf G}-w_{j,\mathsf B}}{\sqrt{1-p_j}}$. Then,
\begin{align*}
 \|x_j\|^2 &=\left\|\frac{w_{j,\mathsf G}-w_{j,\mathsf B}}{\sqrt{1-p_j}}\right\|^2=\frac{\|w_{j,\mathsf G}\|^2+\|w_{j,\mathsf B}\|^2}{1-p_j}
 =1,
 \end{align*}
 and
\begin{align*} \ip{x_j}{x_k}
 &=-\frac{1-p_j-p_k}{\sqrt{(1-p_j)(1-p_k)}}=-\frac{(1-4q)/3}{(2-2q)/3}+O(h)
=-\frac{1-4q}{2-2q}+O(h).
\end{align*}
Its parameter and rounding loss are given here by $\kappa_-(q)=1+\frac{2-2q}{1-4q}
 =\frac{3-6q}{1-4q}$ and $\frac{\kappa_-(q)-2}{\kappa_-(q)}=\frac{(3-6q)-2(1-4q)}{3-6q}
 =\frac{1+2q}{3-6q}$, where all denominators are bounded below uniformly for
$q\in[-1/2,0]$, and the limiting parameter at $q=-1/2$ is two.

For $0\le q\le1/4$, project off the root. On an edge $jk$, let 
$x_j:=\frac{v_j-q_jv_i}{\sqrt{1-q_j^2}}$ where $\|x_j\|^2=1$. Then,
\begin{align*}
 \ip{x_j}{x_k}
 &=\frac{\ip{v_j}{v_k}-q_jq_k}
         {\sqrt{(1-q_j^2)(1-q_k^2)}}=-\frac{1/2+q^2}{1-q^2}+O(h)
         \end{align*} and we have that 
$\kappa_0(q)=1+\frac{1-q^2}{1/2+q^2}
 =\frac{3}{1+2q^2}$ and $\frac{\kappa_0(q)-2}{\kappa_0(q)}
 =1-\frac{2(1+2q^2)}{3}
 =\frac{1-4q^2}{3}$.

For $1/16\le q\le1/4$, we also use the construction of
\cite[Lemma~3.4]{BHL}. Assign signs $+1$ to $\mathsf R$ and $-1$ to
the other two colors, and set
$y_j:=w_{j,\mathsf R}-w_{j,\mathsf G}-w_{j,\mathsf B}$.
Orthogonality at each vertex gives $\|y_j\|^2=1$.
On an edge, exactly one endpoint has color $\mathsf R$ with conditional
probability $p_j+p_k$, and both avoid it with probability $1-p_j-p_k$.
Therefore
\begin{align*}
 \ip{y_j}{y_k}
 &=(1-p_j-p_k)-(p_j+p_k)\\
 &=1-2p_j-2p_k
 =-\frac{1+8q}{3}+O(h),
 \end{align*} with
$\kappa_+(q)=1+\frac{3}{1+8q}
 =\frac{4+8q}{1+8q}$, and 
$\frac{\kappa_+(q)-2}{\kappa_+(q)}=\frac{(4+8q)-2(1+8q)}{4+8q}
 =\frac{1-4q}{2+4q}$. Now, apply KMS to each available construction and keep the larger returned
set. A construction with loss $\ell$ gives
$\E|I|\ge\widetilde\Omega(m\Delta^{-\ell-O(h)})
=m\Delta^{-\ell-o(1)}$. Choosing the smaller available loss gives
the first three branches of \eqref{eq:rho}.

For $q\ge1/4$, the same $y_j$ construction gives
\begin{align*}
 \ip{y_j}{y_k}
 &=-\frac{1+8q}{3}+O(h)\le-1+O(h),\end{align*} with $
 \kappa=2+O(h)$ and $\frac{\kappa-2}{\kappa}=O(h)$. 
 This yields $\E|I|\ge m\Delta^{-o(1)}$.
If every $q_j>1/4$, the local constraint instead shows that $H$ itself
is independent: an edge would require both
$q_j+q_k>1/2$ and $q_j+q_k\le1/2$.

If moment residuals are bounded by $\xi$, the bound remains uniform. Conditioning divides
by $1/3+O(\xi)$, and every further denominator used above has a fixed
positive lower bound in its stated range. After normalization, the
edge correlations and vector parameters therefore change by
$O(h+\xi)$. For each construction this gives
\begin{align*}
 \kappa_{\rm used}=\kappa(q)+O(h+\xi) \implies \frac{\kappa_{\rm used}-2}{\kappa_{\rm used}}=\frac{\kappa(q)-2}{\kappa(q)}+O(h+\xi),\end{align*} and so
$\E|I|\ge m\Delta^{-\rho(q)-o(1)}$.
Near parameter two, choose a valid parameter at least
$2+O(h+\xi)$, and at exactly two, use bipartite coloring. The losses agree at $q=0$ and $q=1/4$, and the projection loss gives the minimum on both sides of $q=1/16$. So, the same estimate holds
uniformly for buckets crossing the boundaries.\looseness=-1
\end{proof}

\paragraph{Pruning.}
The modified global KMS routine samples a standard Gaussian $g$ and
selects $T_g:=\{i:\ip{g}{v_i}\ge t\}$. It computes a maximal matching
$\mathcal M_g$ in $G[T_g]$ and returns the unmatched vertices $U_g$.
Maximality makes $U_g$ independent, and
$|U_g|=|T_g|-2|\mathcal M_g|$ and 
$\E|T_g|=\sum_{i\in V}\Prb[\ip{g}{v_i}\ge t]=n\phib(t)$. If $\E|U_g|\ge n\phib(t)/4$, the sparse progress requirement is met.
Otherwise fix a constant $c_{\rm pr}>c_*$, sufficiently close to $c_*$,
and put $d_{\rm pr}=c_{\rm pr}/(1+c_{\rm pr})$. 

The buffered pruning
argument in Appendix~\ref{app:pruning} gives a nonempty retained
graph $G'$ with:
\begin{enumerate}
\item an outgoing $(\sqrt3t-o(t),t^{-O(1)})$-cover supported on
$\{v_{ij}:j\in N_{G'}(i)\}$ at every retained vertex;
\item for every retained edge $ij$, an outgoing two-walk subcover at
$j$ whose parameter
$a=\ip{v_{ji}}{v_{jk}}$ lies in $[-o(1),d_{\rm pr}+o(1)]$.
\end{enumerate}
Every outgoing family has at most $\Delta$ members, with all supports
in the same graph $G'$. An endpoint reached by a lift therefore retains
an outgoing cover for the next lift. The variance bound still uses
$c_*$. We certify initialization on $a\in[0,d]$; its strict margin
permits the fixed enlargement to $d_{\rm pr}$ in
Appendix~\ref{app:buffer}, while the spread theorem uses the fixed
parameter $c_{\rm pr}$.

\paragraph{Shell size and cover threshold.}
If a full root-correlation shell $H$ has at least
$\Delta^{S+\rho(q)}$ vertices, shell rounding gives
\begin{align*}
 \E|I|
 &\ge |H|\Delta^{-\rho(q)-o(1)}
 \ge\Delta^{S-o(1)}\\
 &=\Delta^{R_*}\Delta^{1/10000-o(1)}
 \\
 &=n^{1-\alpha}\Delta^{1/10000-o(1)},
\end{align*}
which exceeds the sparse target by a fixed polynomial factor.
Hence assume each full shell has fewer than $\Delta^{S+\rho(q)}$
vertices, up to bucket variation. Any endpoint subcover $J$ then satisfies $|J|\le\min\{n,\Delta^{S+\rho(q)+o(1)}\}$, so 
$D\log_\Delta|J|
 \le\min\{D\log_\Delta n,\ D(S+\rho(q))\}+o(1)$,
and $D\log_\Delta n=\frac D\beta=5+3c_*$.
For a cover of root directions at correlation $q$ and threshold $Lt$,
we may consequently use the cardinality exponent
\begin{equation}\label{eq:K}
 K(q):=\min\{5+3c_*,\ D(S+\rho(q))\}.
\end{equation}
The first term uses only the total number of vertices.
For unit directions $(p_j)_{j\in J}$ with cover probability
$\delta\ge t^{-C}$,
\begin{align*}
 t^{-C}
 &\le\delta
 \le\Prb\!\left[\bigcup_{j\in J}\{\ip{g}{p_j}\ge Lt\}\right]\\
 &\le\sum_{j\in J}\Prb[\ip{g}{p_j}\ge Lt]
 =|J|\phib(Lt)\\
 &\le\Delta^{K(q)/D+o(1)}e^{-L^2t^2/2}
 \\
 &=\exp\!\left(\frac{K(q)-L^2}{2}t^2+o(t^2)\right),
\end{align*}and $L^2\le K(q)+\frac{2C\log t}{t^2}+o(1)$.
It follows then, that every such cover satisfies $L^2\le K(q)+o(1)$,
unless an expected-output alternative already holds.

\subsection{The repeated neighborhood lift}\label{sec:lift}
\begin{lemma}\label{lem:lift}
Let $L>0$. Suppose retained endpoints $j$ form a cover of root directions
$p_j=(v_j-q_jv_i)/\sqrt{1-q_j^2}$ at threshold $Lt$, where
$q_j=\ip{v_i}{v_j}$ belongs to a bucket represented by
$q\in[-1/2,3/5]$. Under the size assumptions above, some retained
successor shell has correlation $z$ satisfying
\begin{equation}\label{eq:feasible}
 z\ge-\tfrac12,\quad q+z\le\tfrac12,\quad
 (2z+q)^2\le3d(1-q^2),
\end{equation}
up to vanishing errors, and forms a cover at threshold
$(M(q,L,z)-o(1))t$, where
\begin{equation}\label{eq:lift}
 M(q,L,z)=
 \frac{(1+2qz)L
 -\sqrt{K(q)-L^2}\sqrt{3d(1-q^2)-(2z+q)^2}
 +3\sqrt{1-q^2}}
 {2\sqrt{1-q^2}\sqrt{1-z^2}}.
\end{equation}
Here \eqref{eq:lift} is used on continuing states $L^2\le K(q)$;
positive parts and vanishing slack cover boundary perturbations.
\end{lemma}
\begin{proof}
Use actual endpoint and successor correlations $q,z$ until the
bucketing step. Put $s=\sqrt{1-q^2}$ and $p=p_j$, so
$v_j=qv_i+sp$. The reverse root direction
$v_{ji}:=(v_i-qv_j)/s=sv_i-qp$ is a unit vector perpendicular to
$v_j$, whether or not $i$ and $j$ are adjacent. For a neighbor $k$ of $j$ in $G'$, write
$r:=v_{jk}$ and $z:=\ip{v_i}{v_k}$. Since
$v_k=-v_j/2+\sqrt3r/2$ and $r\perp v_j$, we obtain $\ip{v_i}{r}=\frac{2z+q}{\sqrt3}$, with
\begin{align*}
 \delta:=\ip{r}{v_{ji}}
 &=\frac{\ip{r}{v_i}-q\ip{r}{v_j}}s
 =\frac{2z+q}{\sqrt3s},
 \end{align*} and
 $u:=r-\delta v_{ji}$ with $\|u\|^2=1-\delta^2$ and $
 \ip{u}{v_j}=0, \ip{u}{v_{ji}}=0$.
 
The pairs $(v_j,v_{ji})$ and $(v_i,p)$ span the same plane, so
$u\perp v_i,p$. Substituting $r=\delta(sv_i-qp)+u$ yields
\begin{align*}
 v_k
 &=-\frac12(qv_i+sp)
   +\frac{\sqrt3}{2}\bigl(\delta sv_i-\delta qp+u\bigr)\\
 &=\left(-\frac q2+\frac{\sqrt3\delta s}{2}\right)v_i
   -\left(\frac s2+\frac{\sqrt3\delta q}{2}\right)p
   +\frac{\sqrt3}{2}u,\\
\end{align*} Moreover, we would have
$-\frac q2+\frac{\sqrt3\delta s}{2}=z$ and $\frac s2+\frac{\sqrt3\delta q}{2}
 =\frac{s^2+q(2z+q)}{2s}
 =\frac{1+2qz}{2s}$. Hence
\begin{equation}\label{eq:geometry}
 v_k-zv_i=-\frac{1+2qz}{2s}p+\frac{\sqrt3}{2}u.
\end{equation}

Every $j$ retains an outgoing cover in $G'$. For a fixed $j$,
let $F_j(g):=\max_{k\in N_{G'}(j)}\ip{g}{v_{jk}}$ and let
$\widehat F_j(g):=\max_{k\in N_{G'}(j)}\ip{g}{u_{jk}}$, where
$u_{jk}$ is the projection just defined. Since $|\delta|\le1$,
\begin{align*}
 \widehat F_j(g)&\ge F_j(g)-|\ip{g}{v_{ji}}|,\\
 \Prb\!\left[\widehat F_j(g)\ge\sqrt3t-o(t)-\log t\right]
 &\ge\Prb\!\left[F_j(g)\ge\sqrt3t-o(t)\right]
       -\Prb[|\ip{g}{v_{ji}}|>\log t]\\
 &\ge t^{-C_j}-2e^{-(\log t)^2/2}
 =t^{-O(1)}.
\end{align*}
Pruning bounds the $C_j$ uniformly. The estimate applies separately to
each parent, and $\log t=o(t)$ leaves threshold $\sqrt3t-o(t)$.
Select a heavy $z$ bucket in each family, then a heavy outer class with
the same selected bucket. Each inner family still has at most $\Delta$
vectors and cover probability $t^{-O(1)}$. We can now replace the
correlations by bucket representatives, with the uniform errors proved
above. Writing $\widehat F_j$ for the selected maximum,
\eqref{eq:normvariance} gives $0\le\Var(\widehat F_j)\le1-\delta^2-\frac1{1+c_*}+o(1)
 =d-\delta^2+o(1)$ where $\delta^2\le d+o(1)$, and $(2z+q)^2\le3d(1-q^2)+o(1)$.
 
The remaining constraints $z\ge-1/2$ and $q+z\le1/2$ come from
the local triple distribution. Together these are \eqref{eq:feasible}.

By Gaussian symmetry, the restricted outer family $(-p_j)$ is also a
cover. Its size is at most $\Delta^{K(q)/D+o(1)}$, and its selected
inner family lies in $p_j^\perp$. Lemma~\ref{lem:compose} therefore
provides a joint witness, on an event of probability $t^{-O(1)}$, with $ \ip{g}{-p}\ge Lt$, $\ip{g}{u}\ge\left(\sqrt3-
          \sqrt{d-\delta^2}\sqrt{K(q)-L^2}-o(1)\right)t$, and $d-\delta^2=\frac{3d(1-q^2)-(2z+q)^2}{3(1-q^2)}$.
          
Positive parts and vanishing slack are understood at the boundaries.
For feasible successors, \eqref{eq:domain} below gives
$q,z\in[-1/2,3/5]$, and consequently $1+2qz\ge2/5>0$.
Both favorable evaluations can therefore be inserted into
\eqref{eq:geometry}. With
$p_k:=(v_k-zv_i)/\sqrt{1-z^2}$, this gives
\begin{align*}
 \frac{\ip{g}{p_k}}t
 &=\frac{1+2qz}{2s\sqrt{1-z^2}}\,\frac{\ip{g}{-p}}t
   +\frac{\sqrt3}{2\sqrt{1-z^2}}\,\frac{\ip{g}{u}}t\\
 &\ge\frac{(1+2qz)L+3s
       -\sqrt3s\sqrt{d-\delta^2}\sqrt{K(q)-L^2}}
       {2s\sqrt{1-z^2}}-o(1)\\
 &=\frac{(1+2qz)L+3\sqrt{1-q^2}
       -\sqrt{3d(1-q^2)-(2z+q)^2}\sqrt{K(q)-L^2}}
       {2\sqrt{1-q^2}\sqrt{1-z^2}}-o(1)\\
 &=M(q,L,z)-o(1).
\end{align*}
Therefore, we have shown that the successor shell has the required cover. Its direction
depends only on the root and endpoint, so repeated paths do not increase
the number of directions beyond the shell size, and all the successors remain
in $G'$. Since we established feasibility of the norm before normalization, the bound below keeps $z$ a fixed distance from $1$.
\end{proof}

\paragraph{The correlation domain.}
The ellipse constraint bounds the successor correlation by
$f(q):=(-q+\sqrt{3d(1-q^2)})/2$. To maximize $f$ on
$[-1/2,1]$, observe that its unique interior critical point on
$(-1,1)$ satisfies
\begin{align*}
 f'(q)&=-\frac12-\frac{\sqrt{3d}\,q}{2\sqrt{1-q^2}},\end{align*} with
 $f'(q)=0
 \iff q=-\frac1{\sqrt{1+3d}}<-\frac12$.
Here $d<1$ makes the critical point smaller than $-1/2$.
Thus $f$ is decreasing on $[-1/2,1]$, and
\begin{equation}\label{eq:domain}
 z\le\frac{-q+\sqrt{3d(1-q^2)}}2
 \le\frac14+\frac34\sqrt d<\frac35.
\end{equation}
Here $d=250001/1250001<49/225$, so
$1/4+(3/4)\sqrt d<1/4+(3/4)(7/15)=3/5$.
Thus $q,z\in[-1/2,3/5]$ and
$\sqrt{1-q^2},\sqrt{1-z^2}\ge4/5$ throughout the recursion.

\subsection{Initialization and uniform drift}\label{sec:drift}
First analyze a two-walk bucket with $a\in[0,d]$. The extension to
the buffered interval given by pruning is proved in
Appendix~\ref{app:buffer}. Define
\begin{equation}\label{eq:eta}
 b=\frac14+\frac34a,\qquad
 \eta(a)=\frac{9-3a-6\sqrt{c_*}\sqrt{d-a^2}}
 {4\sqrt{1-b^2}}.
\end{equation}
For a root $i$ and two-walk $i,j,k$ in a retained bucket of type $a$,
write $v_{jk}=a v_{ji}+u$, where
$u\perp\operatorname{span}\{v_i,v_j\}$ and $\|u\|^2=1-a^2$.
The edge decompositions give $v_j=-\frac12v_i+\frac{\sqrt3}{2}v_{ij}$ and $v_{ji}=\frac{\sqrt3}{2}v_i+\frac12v_{ij}$. So,
and\begin{align*}
 v_k
 &=-\frac12v_j+\frac{\sqrt3}{2}(a v_{ji}+u)\\
 &=\left(\frac14+\frac34a\right)v_i
   -\frac{\sqrt3}{4}(1-a)v_{ij}+\frac{\sqrt3}{2}u.
\end{align*}
The reflected outer directions $-v_{ij}$ form a cover at threshold
$\sqrt3t-o(t)$, with at most $\Delta$ directions. Apply
Lemma~\ref{lem:compose} uniformly with outer parameter
$L_t=\sqrt3-o(1)$, $A=1$, and inner squared norm $r^2=1-a^2$.
At $L=\sqrt3$, the loss factors are
$DA-L^2=3(1+c_*)-3=3c_*$ and
$r^2-1/(1+c_*)=1-a^2-1/(1+c_*)=d-a^2$.
With probability $t^{-O(1)}$, the joint witness satisfies $\ip{g}{-v_{ij}}\ge\sqrt3t-o(t)$ and $\ip{g}{u}\ge
 \left(\sqrt3-\sqrt{3c_*}\sqrt{d-a^2}\right)t-o(t)$. Thus, we have 
 \begin{align*}
 \ip{g}{v_k-bv_i}
 &=\frac{\sqrt3}{4}(1-a)\ip{g}{-v_{ij}}
     +\frac{\sqrt3}{2}\ip{g}{u}\\
 &\ge\left(\frac34(1-a)+\frac32
       -\frac32\sqrt{c_*}\sqrt{d-a^2}\right)t-o(t),
       \end{align*} which implies \begin{align*}
 \left\langle{g},\frac{v_k-bv_i}{\sqrt{1-b^2}}\right\rangle
 &\ge
 \frac{9-3a-6\sqrt{c_*}\sqrt{d-a^2}}
      {4\sqrt{1-b^2}}\,t-o(t)
 \\
 &=\eta(a)t-o(t).
\end{align*}
Notably, the bucket correlations differ from $a$ by $o(1)$, contributing only the
$o(t)$ errors from Section~\ref{sec:gaussian}, and the bounds
$d-a^2\ge d-d^2=d(1-d)>0$ and
$1/4\le b\le1/4+3d/4<3/5$ hold uniformly.

Since $b\ge1/4$, its shell has $\rho(b)=0$. If
$\eta(a)^2>DS$, its cardinality cap is contradicted and rounding
succeeds. Otherwise apply the general lift, with $K(b)=DS$.
At most two generic lifts suffice for the initialization below.

\begin{figure}[t]
\centering
\begin{tikzpicture}[>=stealth]

\begin{scope}[x=1.15cm,y=1.15cm]
  \draw[->, thick] (0,0) -- (0,3.05) node[above] {$v_i$};
  \draw[->, black!45] (0,0) -- (3.65,0)
        node[right] {$v_i^\perp$};
  \draw (-0.04,1.00) -- (0.04,1.00) node[left=3pt] {\scriptsize $q$};
  \draw (-0.04,2.15) -- (0.04,2.15) node[left=3pt] {\scriptsize $z$};
  \draw[dashed, black!45] (0,1.00) -- (3.45,1.00);
  \draw[dashed, black!45] (0,2.15) -- (2.70,2.15);
  \node[black!60, anchor=west] at (0.10,1.20) {\scriptsize shell $q$};
  \node[black!60, anchor=west] at (0.10,2.35) {\scriptsize shell $z$};
  \draw[->, black!45] (0,0) -- (3.00,0.92);
  \draw[->, black!45] (0,0) -- (3.40,1.08);
  \node[black!60] at (2.95,0.48) {\scriptsize cover $\{p_j\}$};
  \draw[->, very thick, blue!70!black] (0,0) -- (3.20,1.00);
  \node[blue!70!black, anchor=west] at (3.43,1.00) {$v_j$};
  \draw[->, very thick, blue!70!black] (0,0) -- (2.30,2.15);
  \node[blue!70!black, anchor=south] at (2.30,2.24) {$v_k$};
  \draw[->, thick, red!70!black] (3.23,1.12)
        to[bend left=14] (2.40,2.12);
  \node[red!70!black, align=left, anchor=west] at (3.35,1.75)
        {\scriptsize one edge\\[-1pt] \scriptsize $jk$ in $G'$};
  \node at (1.85,-0.62) {\small (a) one lift};
\end{scope}

\begin{scope}[xshift=9.1cm,x=1.15cm,y=1.15cm]
  \draw[->] (-1.65,0) -- (2.15,0) node[right] {$q$};
  \draw[->] (0,-0.06) -- (0,3.62) node[above] {$L$};
  \draw (-1.25,0.03) -- (-1.25,-0.03)
        node[below] {\scriptsize $-\frac12$};
  \draw (0.62,0.03) -- (0.62,-0.03)
        node[below] {\scriptsize $\frac14$};
  \draw (1.50,0.03) -- (1.50,-0.03)
        node[below] {\scriptsize $\frac35$};
  \draw[dashed, thick] plot[smooth]
        coordinates {(-1.25,3.10) (-0.60,3.17) (0.15,3.07)
                     (0.75,2.91) (1.55,2.65)};
  \node[align=left, anchor=west] at (1.55,2.64)
        {\scriptsize cap\\[-1pt] \scriptsize $L=\sqrt{K(q)}$};
  \filldraw (0.75,1.60) circle (0.9pt)
        node[right=3pt] {\scriptsize $(b,\eta(a))$};
  \draw[->, thick] (0.75,1.60) -- (0.02,1.91);
  \draw[->, thick] (0.02,1.91) -- (1.10,2.23);
  \draw[->, thick] (1.10,2.23) -- (-0.50,2.55);
  \draw[->, thick] (-0.50,2.55) -- (0.50,2.87);
  \draw[->, thick] (0.50,2.87) -- (-0.15,3.23);
  \filldraw[red!70!black] (-0.15,3.23) circle (0.9pt);
  \node[red!70!black, align=right, anchor=east] at (-0.30,3.43)
        {\scriptsize cap violated\\[-1pt] \scriptsize rounding succeeds};
  \node[align=left, anchor=west] at (0.10,0.67)
        {\scriptsize each lift:\\[2pt]
         \scriptsize $L\to M(q,L,z)\ge L+\frac1{1250}$};
  \node at (0.25,-0.62) {\small (b) the recursion};
\end{scope}

\end{tikzpicture}
\caption{Schematic of the recursion. (a)~One lift: a neighbor $k$ of
the endpoint $j$ in the retained graph $G'$ carries the cover from
shell $q$ to a successor shell at correlation $z$, and the
root-orthogonal component $\sqrt{1-z^2}\,p_k$ inherits the lifted
Gaussian threshold. (b)~The scalar recursion on states $(q,L)$: each
lift increases the threshold by a uniform amount until the cardinality
cap $L^2\le K(q)$ is violated, at which point shell rounding returns a
large independent set.}
\label{fig:recursion}
\end{figure}
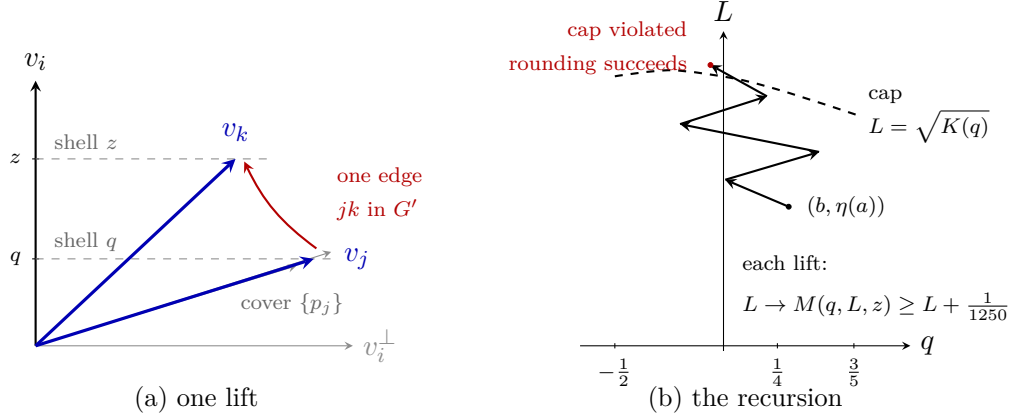

\begin{proposition}[Finite inequalities]\label{prop:finite}
For the rational parameter $c_*$ in \eqref{eq:parameters}, put
$T=53/25$. The following statements hold.
\begin{enumerate}
\item For every $a\in[0,d]$, begin at $(b,\eta(a))$.
Along every feasible sequence, a terminal cap $L^2\le K(q)$ is
violated or a threshold strictly above $T$ is reached within at most
two generic lifts.
In particular, if the seed continues below $T$ and
$L_1=M(b,\eta(a),z)\le T$ is a continuing first successor,
every feasible second successor $w$ satisfies $M(z,L_1,w)>T$.
\item For every feasible continuing state
$q,z\in[-1/2,3/5]$ with $T\le L\le\sqrt{K(q)}$, we have $M(q,L,z)-L>1/1000$.
This drift statement holds uniformly for
$c\in[1/4,1/4+10^{-6}]$, using the corresponding $d,D,S,K$.
\end{enumerate}
\end{proposition}
\begin{proof}
We verify both assertions by rational interval arithmetic on the compact
domain. The verifier and its recorded output are available at
\url{https://github.com/emiletimothy/3coloringVerifier.git}, and the accompanying README gives the command needed to reproduce the verification.

Now, set $Q=2^{120}$ and round rational operations outward using
$\operatorname{down}(x)=Q^{-1}\lfloor Qx\rfloor$ and
$\operatorname{up}(x)=Q^{-1}\lceil Qx\rceil$, so
$\operatorname{down}(x)\le x\le\operatorname{up}(x)$.
For $x\ge0$, compute $k=\lfloor\sqrt{\lfloor Q^2x\rfloor}\rfloor$
by integer square root. Then $k^2\le\lfloor Q^2x\rfloor\le Q^2x<(k+1)^2$ and $\frac{k}{Q}\le\sqrt{x}<\frac{k+1}{Q}$. If $k^2=Q^2x$, use $k/Q$ for both endpoints. Apply these rules
outward to interval endpoints. A positive lower endpoint certifies
positivity on the whole box; otherwise subdivide.

For initialization, the domain $[0,d]\times[-1/2,3/5]$ in the seed
parameter $a$ and the first successor correlation $z$ is covered by an
$8\times8$ grid; boxes violating a necessary feasibility constraint
are discarded, and unresolved boxes are bisected adaptively. Boxes
proving an immediate terminal violation or a first threshold above
$T$ are accepted directly. The remaining boxes are handled by the
following uniform second-lift bound.\\

For a state $(q,L)$, put $s=\sqrt{1-q^2}$, $A=\sqrt{3d}\,s$,
$W=\sqrt{K(q)-L^2}$, and $u=2z+q$. The denominator is
$V=2s\sqrt{1-z^2}=s\sqrt{4-(u-q)^2}$.
Feasibility gives $u^2\le A^2$ and $0<V\le2$, and
\begin{align*}
 H:=V\bigl(M(q,L,z)-T\bigr)
 &= (1+2qz)L+3s-W\sqrt{A^2-u^2}-TV\\
 &= s^2L+3s+qLu-W\sqrt{A^2-u^2}
      -Ts\sqrt{4-(u-q)^2}.
\end{align*}
For $\ell\in\R$, write
$qLu=(qL-\ell)u+\ell(u-q)+\ell q$ and apply Cauchy--Schwarz twice:
\begin{align*}
 (qL-\ell)u-W\sqrt{A^2-u^2}
 &=\ip{(qL-\ell,-W)}{(u,\sqrt{A^2-u^2})}\\
 &\ge-\sqrt{(qL-\ell)^2+W^2}
          \sqrt{u^2+(A^2-u^2)}\\
 &=-A\sqrt{(qL-\ell)^2+W^2}\end{align*} and \begin{align*}
 \ell(u-q)-Ts\sqrt{4-(u-q)^2}
 &=\ip{(\ell,-Ts)}{(u-q,\sqrt{4-(u-q)^2})}\\
 &\ge-\sqrt{\ell^2+T^2s^2}
          \sqrt{(u-q)^2+4-(u-q)^2}\\
 &=-2\sqrt{\ell^2+T^2s^2}.
\end{align*}
Combining these inequalities gives
\begin{equation}\label{eq:dual}
\begin{aligned}
 H\ge{}&s^2L+3s+\ell q
       -A\sqrt{(qL-\ell)^2+W^2}\\
       &\hspace{28mm}-2\sqrt{\ell^2+T^2s^2}.
\end{aligned}
\end{equation}

The bound holds for every feasible successor and every $\ell$.
On each box choose any rational $\ell$ and verify the bound by exact
interval arithmetic. A poor choice may require further subdivision,
but cannot cause a false acceptance.

Before this check, replace the first-lift threshold by its positive
dyadic lower endpoint. This is valid because, with
$B=\sqrt{3d(1-q^2)-(2z+q)^2}$ and
$V=2\sqrt{1-q^2}\sqrt{1-z^2}$, we have that for all $0\leq L < \sqrt{K(q)}$,
\begin{align*}
 \frac{\partial M}{\partial L}(q,L,z)
 &=\frac{1+2qz+B L/\sqrt{K(q)-L^2}}{V}>0,
\end{align*}
where $1+2qz\ge2/5$, and continuity extends monotonicity to the
endpoint. Thus $0<\underline L_1\le L_1$ implies
$M(z,L_1,w)\ge M(z,\underline L_1,w)>T$ for every feasible second
successor $w$. Subdivision terminates with smallest threshold margin
at least $2\times10^{-7}$.
For drift, retain $s,B,V$ above and set $a_0=1+2qz-V$, $K=K(q)$.
Then
\begin{align*}
 H(L):=V\bigl(M(q,L,z)-L\bigr)
 =a_0L-B\sqrt{K-L^2}+3s,
 \end{align*} and $H'(L)=a_0+\frac{BL}{\sqrt{K-L^2}}$ and $H''(L)=\frac{BK}{(K-L^2)^{3/2}}\ge0$, for $T\le L  < \sqrt K$.
 
In particular, when $K>T^2$, convexity gives the supporting-line bound
\begin{align*}
 H(L)&\ge H(T)+H'(T)(L-T)\\
 &\ge H(T)+\min\{0,H'(T)\}(\sqrt K-T)
 \qquad(T\le L\le\sqrt K).
\end{align*}
A second, global lower bound follows from Cauchy--Schwarz:
\begin{align*}
 H(L)-3s
 &=\ip{(a_0,-B)}{(L,\sqrt{K-L^2})}\\
 &\ge-\sqrt{a_0^2+B^2}\sqrt{L^2+K-L^2}
 =-\sqrt K\sqrt{a_0^2+B^2}.
\end{align*}
Thus the verifier may use the larger of these two lower bounds.
When $H'(T)\ge0$, the first reduces to $H(L)\ge H(T)$.
When $K=T^2$, the interval consists of the single point $L=T$,
which is evaluated directly; when $K<T^2$, no continuing state is
present. Boxes where a denominator enclosure contains zero are
handled by direct endpoint bounds or subdivision. These inequalities
treat $L$ universally, leaving only $q,z$ as subdivision variables.
A finite adaptive subdivision of the $(q,z)$-domain verifies
$H(L)-V/1000>10^{-7}$ on every nonempty accepted box, with smallest
numerator margin at least $2\times10^{-5}$. To enclose $K$, the function
$\rho$ is evaluated at the interval endpoints and at zero when present;
$\rho$ is continuous, increases up to $q=0$, and decreases thereafter.
Both expressions defining $K$ are increasing in $c$, so endpoint
evaluation also encloses the entire stated parameter interval. This
proves both assertions.

The initialization and drift checks terminate with $8{,}140$
and $355$ final boxes, respectively, including boxes excluded
by necessary feasibility constraints.
The certified initialization threshold margin exceeds
$2.55\times10^{-7}$, and the certified lower bound on
$V(M(q,L,z)-L-1/1000)$ exceeds $2.83\times10^{-5}$.
Both checks cover their full stated domains and leave no
unresolved boxes.
\end{proof}

\subsection{Termination, algorithm, and exponent}\label{sec:algorithm}
Suppose neither global KMS nor any root-correlation shell has the
required expected output. Pruning,
\eqref{eq:eta}, and Proposition \ref{prop:finite}(1), using at most two generic lifts, produce an
endpoint cover with $L\ge53/25$ for all sufficiently large $n$; its
strict verified margin absorbs the negligible initialization errors.

At any continuing state, we have that the exact drift
margin and continuity allow the uniform lower increment $1/1250$ after
all negligible errors have been included. This remains valid close to
the ellipse boundary and close to $L^2=K$: positive-part square roots
are continuous, and their perturbations vanish uniformly. Near a
slightly infeasible boundary one can pass to a nearest feasible point,
or equivalently use compactness and the strict margin. Thus each lift
either contradicts the shell size bound or increases $L$ by at least
$1/1250$.

After $j$ continuing lifts, the thresholds satisfy $L_{j+1}\ge L_j+\frac1{1250}$, where $L_j\ge L_0+\frac{j}{1250}
       \ge\frac{53}{25}+\frac{j}{1250}$, and $L_{400}\ge\frac{53}{25}+\frac{400}{1250}
          =\frac{61}{25}=2.44$.
On the other hand, the endpoint family $\mathcal P_j$ contains at
most $n$ distinct unit vectors and has cover probability at least
$t^{-C_j}$ for a constant $C_j$. Using
$\log n=(5+3c_*)t^2/2+o(t^2)$, its upper bound is
\begin{align*}
 t^{-C_j}
 &\le\Prb\!\left[\max_{p\in\mathcal P_j}\ip{g}{p}
                 \ge(L_j-o(1))t\right]\\
 &\le\sum_{p\in\mathcal P_j}
          \Prb[\ip{g}{p}\ge(L_j-o(1))t]\\
 &\le n\exp\!\left(-\frac{(L_j-o(1))^2t^2}{2}\right)\\
 &=\exp\!\left(\frac{5+3c_*-L_j^2+o(1)}2t^2\right).
\end{align*}
Taking logarithms and dividing by $t^2/2$ yields
\begin{align*}
 L_j^2&\le5+3c_*+o(1)+\frac{2C_j\log t}{t^2}
       =5+3c_*+o(1),
       \end{align*} and so
$L_j\le\sqrt{5+3c_*}+o(1)<2.4$, for all sufficiently large $n$.
This contradicts $L_{400}\ge2.44$. Hence one rounding routine has
the required expected output, at an absolute constant recursion depth.

\paragraph{The sparse algorithm.}
Algorithm~\ref{alg:sparse-rounding} tests the global rounding routine
and all $O(n\log n)$ root-correlation shells. In the middle
positive-correlation range it uses both vector constructions from
Lemma~\ref{lem:shell-rounding}. For each shell with edges, a
conservative upper bound on its largest edge inner product gives a
valid vector-coloring parameter. A bucket crossing a formula
boundary uses the weaker endpoint guarantee; continuity preserves
the asymptotic exponent.

\begin{algorithm}[H]
\caption{Sparse independent-set routine}\label{alg:sparse-rounding}
\small
\textbf{Input:} a promised 3-colorable graph $G$ of maximum degree
at most $\Delta=n^\beta$; a repetition budget $R=n^C$ for a
sufficiently large fixed $C$.\par
\textbf{Output:} the largest independent set found.
\begin{enumerate}[leftmargin=1.6em,itemsep=2pt,topsep=3pt]
  \item Solve the three-round SoS relaxation to inverse-polynomial
        accuracy. Set $I_*\gets\varnothing$ and choose $t$ with
        $\phib(t)=\Delta^{-1/D}$.
  \item For every root $i$, partition $V\setminus\{i\}$ into
        correlation buckets of width $1/\log\Delta$.
  \item \textbf{For} $z=1,\dots,R$:
    \begin{enumerate}[label=(\alph*),leftmargin=1.7em,itemsep=1pt,topsep=2pt]
      \item Apply Algorithm~\ref{alg:threshold-matching} to $G$,
            the vectors $(v_j)$, and threshold $t$.
      \item \textbf{For each} root and each nonempty bucket $H$:
            if $G[H]$ is edgeless, take $H$ as a candidate;
            otherwise apply KMS rounding to every applicable vector
            coloring from Lemma~\ref{lem:shell-rounding}.
      \item Check each candidate for independence and replace
            $I_*$ whenever a larger independent set is found.
    \end{enumerate}
  \item \textbf{Return} $I_*$.
\end{enumerate}
\end{algorithm}
The preceding argument guarantees that at least one routine has
expected output at least $B=n^{1-\alpha-o(1)}$. If $Y\in[0,n]$
is its output size and $p=\Prb(Y\ge B/2)$, then
\begin{align*}
 B\le\E Y
 &=\E\bigl[Y\mathbf 1_{\{Y<B/2\}}\bigr]
   +\E\bigl[Y\mathbf 1_{\{Y\ge B/2\}}\bigr]\\
 &\le\frac B2(1-p)+np\\
&\le\frac B2+np,
 \end{align*}
 and $\Prb(Y\ge B/2)\ge\frac{B}{2n}$. For $R\ge(40n/B)\log n$ independent repetitions, the failure
probability is bounded by
\begin{align*}
 \Prb\!\left[\max_{1\le r\le R}Y_r<B/2\right]
 &\le\left(1-\frac{B}{2n}\right)^R\\
 &\le\exp\!\left(-\frac{RB}{2n}\right)
 \le n^{-20}.
\end{align*}
The fixed polynomial budget in Algorithm~\ref{alg:sparse-rounding}
can exceed this bound, so the algorithm need not know the hidden factor
in $B$. Its output is at least the output of the successful routine.
The full shells contain the endpoint subcovers used in the proof;
the covers, retained graph, and path witnesses need not be constructed.

The analysis uses an exact SoS solution; the implementation uses weak
feasibility to inverse-polynomial accuracy. Bound moment variables in
$[-2,2]$, strictly containing their exact range $[-1,1]$. Relax every linear equality and
inequality by a tolerance $\xi_n=n^{-C}$, for a sufficiently large
fixed $C$, and the positive semidefinite constraint to a minimum
eigenvalue of at least $-\xi_n$.
Let $x^*$ be an exact feasible moment vector, and let $L_n=n^{O(1)}$
bound the Lipschitz constants of all affine constraint maps and of the
moment-matrix map $\mathcal M$ in operator norm. Such a bound exists
because the relaxation has fixed level and polynomial dimension.
For $\|h\|_2\le\xi_n/(2L_n)$, the equality maps $a_r$, inequality
maps $b_r$ (with $b_r(x^*)\ge0$), and moment matrix obey
\begin{align*}
 |a_r(x^*+h)|
 &\le L_n\|h\|_2\le\xi_n/2,
 \end{align*}
 and \begin{align*}
 b_r(x^*+h)
 &\ge b_r(x^*)-L_n\|h\|_2\ge-\xi_n/2,
 \end{align*}
 Together, this gives 
 \begin{align*}\lambda_{\min}(\mathcal M(x^*+h))
 &\ge\lambda_{\min}(\mathcal M(x^*))
      -\|\mathcal M(x^*+h)-\mathcal M(x^*)\|_{\rm op}\\
 &\ge-L_n\|h\|_2\ge-\xi_n/2.
\end{align*}
For large $n$, this ball also lies within the coordinate box $[-2,2]$.
The relaxed feasible body therefore contains a ball of inverse-polynomial
radius. Standard weak-feasibility methods, with affine-constraint and
eigenvalue separation, find an approximate feasible point in polynomial
time. The 3-colorability promise gives the exact point $x^*$ used
only to establish this interior-ball guarantee.

If $\lambda_{\min}(M)\ge-\xi_n$, then
$\widehat M:=M+\xi_n I\succeq0$ and
$|\widehat M_{ab}-M_{ab}|\le\xi_n$.
Symmetrization preserves the order of these residuals. Conditioning on
a root color is stable because its exact probability is $b=1/3$.
Indeed, for a joint probability $a\in[0,1]$ and approximate values
$\widetilde a=a+O(\xi_n)$, $\widetilde b=b+O(\xi_n)\ge1/6$,
\begin{align*}
 \left|\frac{\widetilde a}{\widetilde b}-\frac ab\right|
 &=\frac{|b(\widetilde a-a)-a(\widetilde b-b)|}
          {b\widetilde b}\\
 &\le\frac{b|\widetilde a-a|+|a|\,|\widetilde b-b|}
           {b\widetilde b}
 =O(\xi_n).
\end{align*}
The local tables are approximate probability tables: nonnegativity,
edge exclusions, and consistency hold up to inverse-polynomial residuals,
rather than becoming exact after the diagonal adjustment.
Edge correlations, triple constraints, and the shell vector
colorings therefore incur only vanishing errors, giving vector
parameters $\kappa+o(1)$. At a formula boundary use the adjacent
limiting bound; near $q=1/4$ this is a vector $(2+o(1))$ coloring.
All actual vectors are normalized before evaluating Gaussian tails.
On the compact correlation domain, these errors are dominated by the
bucket width and produce uniform $o(t)$ threshold errors. The strict
margins of the finite verification and the gap $S-R_*$ absorb them. Thus the exact
SDP convention entails no exact-arithmetic computational assumption.
Appendix~\ref{app:approx} proves the matching-packing and
pruning statements directly for perturbed edge correlations, retaining
exactly symmetric matching weights and exactly orthogonal tangent
directions. Appendix~\ref{app:errors} records the uniform errors and
cover probabilities through depth $404$.

Gaussian sampling and threshold comparisons need only polynomially
many bits. Suppose a call uses at most $n^a$ sampled coordinates and
$n^b$ comparisons, for fixed $a,b$.
For a prescribed failure exponent $h>0$, choose $A>0$ with
$A^2/2>a+h+2$ and truncate each standard Gaussian coordinate at
$A\sqrt{\log n}$. The truncation event satisfies
\begin{align*}
 \Prb[\text{some coordinate is truncated}]
 &\le\sum_{j=1}^{\lfloor n^a\rfloor}
       \Prb[|Z_j|>A\sqrt{\log n}]\\
 &\le2n^a\exp\!\left(-\frac{A^2\log n}{2}\right)\\
 &\le2n^{-h-2}.
\end{align*}
Choose the coordinate and threshold precision so that, on a coupling
with exact samples, the error in each tested difference is at most
$\eta=n^{-C}$, except on an event of probability $n^{-h-2}$.
In dimension $d_n\le n^a$, coordinate error $\delta$ gives
$|\ip{\widetilde g-g}{v}|\le\|\widetilde g-g\|\,\|v\|
\le\sqrt{d_n}\,\delta$. Choose $\delta$ and threshold error
$\eta_\tau$ with $\sqrt{d_n}\,\delta+\eta_\tau\le\eta$. Conditional on the already computed vectors,
each ideal unit-vector evaluation is standard normal. Its probability
of lying within $\eta$ of its threshold $\tau$ is at most
\begin{align*}
 \Prb[|Z-\tau|\le\eta]
 &=\int_{\tau-\eta}^{\tau+\eta}
       \frac{e^{-x^2/2}}{\sqrt{2\pi}}\,\mathrm dx
 \le\frac{2\eta}{\sqrt{2\pi}}.
\end{align*}
Consequently, choosing $C>b+h+2$ gives
\begin{align*}
 \Prb[\text{some rounding decision changes}]
 &\le2n^{-h-2}+n^{-h-2}
      +n^b\frac{2n^{-C}}{\sqrt{2\pi}}=O(n^{-h-2}).
\end{align*}
Use a fixed edge ordering when constructing maximal matchings.
When the threshold decisions agree, the matching and the resulting
independent set agree as well. Every returned set is also checked
combinatorially for independence. The exponent $h$ can be chosen to
fit the per-call failure budget below.

\paragraph{Dense--sparse combination.}
The progress framework in Section~2.2 of~\cite{BHL} combines the
sparse independent-set guarantee with dense progress of cost
$n^{o(1)}\sqrt{n/\Delta}$. The parameter choice in
\eqref{eq:parameters} balances these costs explicitly by setting $\Delta^{1/D}
 =n^{\beta/D}
 =n^{1/(5+3c_*)}=n^\alpha$ and 
 \begin{align*}\sqrt{n/\Delta}
 &=n^{(1-\beta)/2}
 \\&=n^{\frac12(1-3(1+c_*)/(5+3c_*))}
 \\&=n^{1/(5+3c_*)} =n^\alpha.
\end{align*}
Both progress costs are therefore $n^{\alpha+o(1)}$.
Choose $\alpha'$ with $\alpha<\alpha'<4/23$. For sufficiently
large $n$, the subpolynomial losses are at most
$n^{\alpha'-\alpha}$, and Theorem~2.6 of~\cite{BHL} gives a
proper coloring with $\widetilde O(n^{\alpha'})$ colors. If the
polylogarithmic overhead is $(\log n)^k$, then
\begin{align*}
 \frac{n^{\alpha'}(\log n)^k}{n^{4/23}}
 &=\exp\!\left(k\log\log n-(4/23-\alpha')\log n\right)
 \longrightarrow0,\\
 \alpha
 &=\frac{1000000}{5750003}
 <\frac4{23}=0.173913043478\ldots.
\end{align*}
For the high-probability guarantee, let $N$ be the original input
size and $N^b$ a polynomial upper bound on the total number of
randomized calls. Amplify each call to failure probability at most
$N^{-b-10}$; then
\begin{align*}
 \Prb[\text{some call fails}]
 &\le\sum_{r=1}^{\lfloor N^b\rfloor}
       \Prb[\text{call }r\text{ fails}]\\
 &\le N^bN^{-b-10}=N^{-10}.
\end{align*}
Together with the color bound, this proves Theorem~\ref{thm:main}.
Finitely many smaller input sizes are absorbed into the constant
in the asymptotic guarantee.

\begin{remark}
The dense case is governed by a minimum-degree condition through the
cost $n^{o(1)}\sqrt{n/\Delta}$. Whether expansion-type structure in
dense graphs can substitute for, or refine, this condition in dense
progress arguments is an interesting direction; rank-expansion, a
well-linkedness measure defined through the cutrank of balanced
separations~\cite{Anand}, is one candidate notion of expansion for
dense graphs.
\end{remark}

\section{Conclusion}
We provide a randomized polynomial time algorithm to color a graph that is promised to be 3-colorable with $O(n^{4/23})$ colors. The proof reduces the sparse coloring problem to a scalar recursion on
root correlation and cover threshold. The variance bound for Gaussian
maxima controls the loss at each step, while shell rounding limits the
number of available endpoints. The resulting positive drift is verified
on a compact domain by exact interval arithmetic. The number of steps
is an absolute constant, and the algorithm rounds full correlation shells.\looseness=-1

\section*{Acknowledgements}

We thank Ijay Narang for discussions that prompted us to study
this problem, for subsequent exchanges about our respective
approaches, and for sharing his concurrent work and preprint with Yukai
Tang \cite{narang2026improvedsdpcoloring3colorable}. This work was supported by NSF Grants CCF 2338816 and CCF 2504994.

\bibliographystyle{alpha}
\bibliography{main}

\newpage 
\appendix

\section{Pruning with a fixed parameter buffer}\label{app:pruning}

This appendix gives the upper-degree version of the pruning argument
used in the proof. It follows the packing and spread framework of
\cite[Theorem 4.11 and Appendix A.1]{BHL}, with vertex and edge deletions
interleaved so that all resulting covers have support in the same final
graph.

\subsection{Packings and the auxiliary parameter}

For unit vectors $(r_j)_{j\in J}$, nonnegative weights $\mu(j)$ form an
$(s,\delta)$-packing if $\sum_{j\in J}\mu(j)\ge\delta$ and, for every
$A\subseteq J$,
\[
 \sum_{j\in A}\mu(j)
 \le \Prb_g\!\left[\max_{j\in A}\ip{g}{r_j}\ge s\right].
\]
The inequality is required for every subset, so it survives restriction.
In particular, restricting to $J'\subseteq J$ gives
\begin{align*}
 \Prb_g\!\left[\max_{j\in J'}\ip{g}{r_j}\ge s\right]
 &\ge \sum_{j\in J'}\mu(j),\\
&=\sum_{j\in J}\mu(j)-\sum_{j\in J\setminus J'}\mu(j).
\end{align*}
Thus deleting mass $b$ leaves a cover of probability at least $\delta-b$.
Repeated directions may be aggregated when applying a theorem stated
for sets of vectors: the new weight of a direction is the sum of the
weights of its copies, and its Gaussian evaluation is unchanged.

Keep the analysis parameter $c=c_*$ fixed and put $D=3(1+c)$.
The degree parameter $\Delta$ is an upper bound for $\Delta(G)$.
For $\phib(t)=\Delta^{-1/D}$,
\begin{equation}\label{eq:app-tail-scale}
 \log\Delta=\frac{3(1+c)}2t^2+O(\log t).
\end{equation}
Choose another fixed parameter $c_{\rm pr}>c$, sufficiently close to $c$
as justified below. For $s=(\sqrt3-o(1))t$, the Gaussian tail asymptotic
and~\eqref{eq:app-tail-scale} give $\log\frac1{\phib(s)}=\frac{s^2}{2}+O(\log s)
   =\frac32t^2+o(t^2)$, so
\begin{align*}
 (1+c_{\rm pr})\log\frac1{\phib(s)}-\log\Delta
 &=\frac32\bigl((1+c_{\rm pr})-(1+c)\bigr)t^2+o(t^2)\\
 &=\frac32(c_{\rm pr}-c)t^2+o(t^2)>0.
\end{align*} So, $|N(i)|\le\Delta
\le\phib(s)^{-(1+c_{\rm pr})}$
for all sufficiently large $n$. Then, via the spread theorem at its fixed parameter, every outgoing family is therefore
at most $c_{\rm pr}$-inefficient at threshold $s$.

\subsection{A retained graph with persistent two-walk supports}

\begin{lemma}\label{lem:app-persistent-pruning}
Suppose an $n$-vertex graph has symmetric nonnegative edge weights
$\mu_i(j)=\mu_j(i)$, with row sums at most one and
\[
 \sum_{\{i,j\}\in E}\mu_i(j)\ge n/8.
\]
For every vertex $i$, suppose unit directions $r_{ij}$ form an outgoing
packing at threshold $s=(\sqrt3-o(1))t$ with weights $\mu_i$.
Assume $|N(i)|\le\Delta\le\phib(s)^{-(1+c_{\rm pr})}$.
Put $\eps=1/\log\log t$, $\tau=\eps^2$, and
\[
 \lambda=\frac{c_{\rm pr}}{1+c_{\rm pr}}
       \left(1+\frac{C_1}{\log s}\right),\qquad
 p=\exp[-C_2(\log s)^2],
\]
where $C_1,C_2>0$ depend only on $c_{\rm pr}$.
For sufficiently large $n$, there is a nonempty retained graph $G'$
such that every retained vertex $i$ has
\begin{equation}\label{eq:app-row}
 \mu_i(N_{G'}(i))\ge1/8-O(\eps),
\end{equation}
and every oriented retained edge $i\to j$ has
\begin{equation}\label{eq:app-twowalk}
 \mu_j\{k\in N_{G'}(j):
       -\eps\le\ip{r_{ji}}{r_{jk}}\le\lambda\}
 \ge\tau/2.
\end{equation}
\end{lemma}

\begin{proof}
Apply \cite[Theorem 4.11]{BHL} at each vertex. It gives
$X_i\subseteq N(i)$ with lost mass at most $1/\log s$ and, for every
$h\in X_i$,
\begin{equation}\label{eq:app-spread}
 \mu_i\{k\in X_i:\ip{r_{ih}}{r_{ik}}\ge\lambda\}\le p.
\end{equation}
Rows with total mass at most $1/\log s$ may be discarded entirely.
Delete an edge unless both endpoints retain it, and let $E_0$ be the
remaining edge set. Symmetry of the weights gives the explicit charge
\begin{align*}
 \sum_{\{i,j\}\in E\setminus E_0}\mu_i(j)
 &\le\sum_i\sum_{j\in N(i)\setminus X_i}\mu_i(j)\\
 &\le\sum_i\frac1{\log s}
   =\frac n{\log s}.
\end{align*}
An edge discarded by both endpoints may be counted twice on the right,
which only increases this upper bound. Every subsequent neighborhood is
a subset of its $X_i$, so~\eqref{eq:app-spread} continues to hold for
every retained neighbor.

Set $B_t=\frac1{\log s}+\tau(1+1/\eps)$ and $r_0=\frac18-2B_t$. For large $n$, $r_0>0$. We repeatedly apply either of the following operations whenever possible:
\begin{enumerate}
\item Delete a vertex whose current incident weight is less than $r_0$.
\item If a current oriented edge $i\to j$ satisfies
 \[
 E_{ij}:=\{k\in N_{\rm current}(j):
                 \ip{r_{ji}}{r_{jk}}\ge-\eps\},
 \qquad \mu_j(E_{ij})<\tau,
 \]
 delete every edge $jk$ with $k\in E_{ij}$.
\end{enumerate}
The second operation removes the triggering edge $ij$ itself, since
$\ip{r_{ji}}{r_{ji}}=1$. If $i_1,\ldots,i_q$ successively trigger that
operation at a fixed center $j$, a later triggering edge survived every
earlier deletion. Consequently, for $a<b$, we have that $i_b\notin E_{i_a j}
 \implies \ip{r_{ji_a}}{r_{ji_b}}<-\eps$, and so
\begin{align*}
 0\le\left\|\sum_{a=1}^q r_{ji_a}\right\|^2
 &=\sum_{a=1}^q\|r_{ji_a}\|^2
   +2\sum_{1\le a<b\le q}\ip{r_{ji_a}}{r_{ji_b}}\\
 &\le q-2\eps\binom q2
   \\
   &=q\bigl(1-\eps(q-1)\bigr).
\end{align*}
For $q\ge1$, division by $q$ yields $q\le1+1/\eps$; the bound also
holds if no operation is triggered at $j$. If $q_j$ is the number of
operations at center $j$, their total deletion charge is therefore
\begin{align*}
 W_{\rm edge}
 &\le\sum_j\sum_{a=1}^{q_j}\mu_j(E_{i_a j})\\
 &\le\sum_j q_j\tau
\\
&\le n\tau(1+1/\eps).
\end{align*}
Here each $E_{i_a j}$ is evaluated in the graph at the time of its
deletion. Charging a vertex deletion only for edges still present at
that time similarly gives $W_{\rm vertex}\le nr_0$.

Every operation deletes a vertex or an edge, so the process terminates.
If no edge remained, all the initial weight would have been deleted,
and the three charges would give
\begin{align*}
 \frac n8
 &\le\sum_{\{i,j\}\in E}\mu_i(j)\\
 &\le\frac n{\log s}+n\tau(1+1/\eps)+nr_0\\
 &=n(B_t+r_0)
   \\
   &=n\left(\frac18-B_t\right)
   <\frac n8,
\end{align*}
which is a contradiction. At termination, every retained vertex satisfies $\mu_i(N_{G'}(i))
 \ge r_0
   =\frac18-\frac2{\log s}-2\eps^2-2\eps
   =\frac18-O(\eps)$.
   For a final oriented edge $i\to j$, the second operation is no longer
applicable, so $\mu_j(E_{ij})\ge\tau$. Also $i\in X_j$, and
\eqref{eq:app-spread} bounds the mass with inner product above $\lambda$.
For sufficiently large $n$, $p\le\tau/2$, whence
\begin{align*}
 \mu_j\{k\in N_{G'}(j):-\eps\le\ip{r_{ji}}{r_{jk}}\le\lambda\} &=\mu_j(E_{ij})
   -\mu_j\{k\in N_{G'}(j):\ip{r_{ji}}{r_{jk}}>\lambda\}\\
 &\ge\tau-p\\ &\ge\frac\tau2.
\end{align*}
This proves~\eqref{eq:app-row} and~\eqref{eq:app-twowalk} in the same
final graph, with all their supporting neighbors still present.
\end{proof}

\subsection{Matching weights and approximate edge correlations}\label{app:approx}

For exact unit vectors with edge inner product $-1/2$, let
$S_g=\{i:\ip{g}{v_i}\ge t\}$ and choose a maximal matching $M$ in
$G[S_g]$. Its unmatched vertices form an independent set $I_g$.
Put $p_t=\phib(t)$ and define
\[\mu_i(j)=\frac{\Prb[\{i,j\}\in M]}{p_t}=\mu_j(i).\]
Then,\begin{align*}
 d_i:=\sum_{j\in N(i)}\mu_i(j)
 &=\frac{\Prb[i\text{ is matched in }M]}{p_t}
   =\Prb[i\text{ is matched in }M\mid i\in S_g]\le1\end{align*}
   So, we have
   \begin{align*}
 \E|I_g|
 &=\sum_i\Prb[i\in S_g\text{ and }i\text{ is unmatched}]\\
 &=p_t\sum_i(1-d_i).
\end{align*}
The sum in the second line is valid because the matching contains at
most one edge incident to $i$. If $\E|I_g|<np_t/4$, these identities give
$\sum_i d_i>3n/4$. In particular, for $H=\{i:d_i\ge1/2\}$,
\begin{align*}
 \frac{3n}{4}<\sum_i d_i
 &\le |H|+\frac{n-|H|}{2}
   =\frac n2+\frac{|H|}{2},
\end{align*} which gives
 $|H|>\frac n2$, thereby implying
 \begin{align*}
 \sum_{\{i,j\}\in E}\mu_i(j)
 &=\frac12\sum_i d_i
 \ge\frac12\sum_{i\in H}d_i
 \ge\frac{|H|}{4}\ge\frac n8.
\end{align*}

For each edge, the direction $r_{ij}=(v_j+v_i/2)/(\sqrt3/2)$ satisfies $\ip{v_i}{r_{ij}}=\frac{-1/2+1/2}{\sqrt3/2}=0$ with $\|r_{ij}\|^2=\frac{\|v_j\|^2+\ip{v_i}{v_j}+\|v_i\|^2/4}{3/4}=1$ where
we have that $\{i,j\}\in M \implies
 \ip{g}{r_{ij}}
 =\frac{\ip{g}{v_j}+\ip{g}{v_i}/2}{\sqrt3/2}
 \ge\sqrt3t$.
 
Now, fix $A\subseteq N(i)$. The matching events indexed by $j\in A$ are
disjoint, and their union implies that $i$ is selected and one of the
directions in $A$ reaches threshold $\sqrt3t$. Since all these
directions lie in $v_i^\perp$, their evaluations are jointly independent
of $\ip{g}{v_i}$. Therefore
\begin{align*}
 \sum_{j\in A}\mu_i(j)
 &=\frac1{p_t}\Prb\!\left[\bigcup_{j\in A}\{\{i,j\}\in M\}\right]\\
 &\le\frac1{p_t}\Prb\!\left[
      \ip{g}{v_i}\ge t,\ \max_{j\in A}\ip{g}{r_{ij}}\ge\sqrt3t
      \right]\\
 &=\frac{\Prb[\ip{g}{v_i}\ge t]}{p_t}
      \Prb\!\left[\max_{j\in A}\ip{g}{r_{ij}}\ge\sqrt3t\right]\\
 &=\Prb\!\left[\max_{j\in A}\ip{g}{r_{ij}}\ge\sqrt3t\right].
\end{align*}
These are precisely the packing inequalities required by
Lemma~\ref{lem:app-persistent-pruning}.

For implementation, normalize all vectors exactly and suppose that
edge correlations satisfy
$e_{ij}:=\ip{v_i}{v_j}=-1/2+O(\xi_n)$, uniformly, where
$\xi_n$ is inverse polynomial. Instead use the exact unit directions $r_{ij}=\frac{v_j-e_{ij}v_i}{\sqrt{1-e_{ij}^2}}$. The matching weights are defined by the same undirected events, so
$\mu_i(j)=\mu_j(i)$ still holds even though the directions at the two
endpoints differ. 

The new directions satisfy exactly $\ip{r_{ij}}{v_i}=\frac{e_{ij}-e_{ij}\|v_i\|^2}{\sqrt{1-e_{ij}^2}}=0$ and $\|r_{ij}\|^2=\frac{1-2e_{ij}^2+e_{ij}^2}{1-e_{ij}^2}=1$.

For large $n$, $e_{ij}<0$. On a matched edge, both endpoints are
selected, and hence
\begin{align*}
 \ip{g}{r_{ij}}
 &=\frac{\ip{g}{v_j}-e_{ij}\ip{g}{v_i}}{\sqrt{1-e_{ij}^2}}\\
 &\ge\frac{(1-e_{ij})t}{\sqrt{1-e_{ij}^2}}
   \\
   &=\sqrt{\frac{1-e_{ij}}{1+e_{ij}}}\,t\\
 &\ge(\sqrt3-C\xi_n)t
\end{align*}
for a uniform constant $C$. Thus the preceding packing calculation
applies at the common threshold $s=(\sqrt3-C\xi_n)t$.
Orthogonality is exact, and the auxiliary parameter buffer gives
$\Delta\le\phib(s)^{-(1+c_{\rm pr})}$. The pruning proof uses only
these packing inequalities, symmetry, unit norms, and the displayed
Gram-matrix calculation, all of which remain valid.

The subsequent geometric identities can also be written before any
approximation. Fix a root $i$, a current endpoint $j$, and a neighbor
$k$ of $j$, and let
 $q=\ip{v_i}{v_j}$, $z=\ip{v_i}{v_k}$, and $e=\ip{v_j}{v_k}$. Also, let
 $\sigma=\sqrt{1-q^2}$,
 $p_j=\frac{v_j-qv_i}{\sigma}$,
 $h=\frac{v_i-qv_j}{\sigma}=\sigma v_i-qp_j$. Then, let $r=\frac{v_k-ev_j}{\sqrt{1-e^2}},
 \delta=\ip{r}{h}
       =\frac{z-eq}{\sqrt{1-e^2}\,\sigma}$, and
 $u=r-\delta h$.
Both $r$ and $h$ are perpendicular to $v_j$, and subtracting
$\delta h$ also makes $u$ perpendicular to $h$. Since
$\operatorname{span}\{v_j,h\}=\operatorname{span}\{v_i,p_j\}$, this
gives $u\perp v_i,p_j$ and $\|u\|^2=1-\delta^2$. Expanding $v_k$ now
yields
\begin{align*}
 v_k
 &=ev_j+\sqrt{1-e^2}\,r\\
 &=e(qv_i+\sigma p_j)
   +\sqrt{1-e^2}\bigl(\delta(\sigma v_i-qp_j)+u\bigr)\\
 &=zv_i+\frac{e-qz}{\sigma}p_j+\sqrt{1-e^2}\,u,
\end{align*}
Therefore, we have
\[v_k-zv_i
 =-\frac{1+2qz}{2\sigma}p_j+\frac{\sqrt3}{2}u
   +O(\xi_n)p_j+O(\xi_n)u,\] by using the fact that $e=-1/2+O(\xi_n)$ and the lower bound on $\sigma$
in the continuing domain. Gaussian evaluations of the error terms are
$o(t)$ with the uniform bounds established below. Approximate local
distributions likewise give $q+z\le1/2+o(1)$, so the exact lift and its
feasibility constraints change only by uniformly vanishing errors.

\subsection{Preservation of the finite initialization bound}\label{app:buffer}

The parameter $c_{\rm pr}$ only changes the allowed initial
two-walk interval to $a\in[-o(1),d_{\rm pr}+o(1)]$ and $d_{\rm pr}=\frac{c_{\rm pr}}{1+c_{\rm pr}}$. All Gaussian variance estimates still use
\eqref{eq:app-tail-scale}, hence the original analysis parameter
$c$, the original $d=c/(1+c)$, and the same functions $K$ and $M$.
In particular the seed formula remains $b=\frac{1+3a}{4}$ and $\eta(a)=\frac{9-3a-6\sqrt c\sqrt{d-a^2}}
                  {4\sqrt{1-b^2}}$.
                  
For $c_{\rm pr}>c$ sufficiently close, the enlarged interval is
strictly below $\sqrt d$, and $b$ is bounded away from one.
The seed formula is therefore continuous on a fixed neighborhood
of the original interval.

Proposition~\ref{prop:finite}(1) asserts a strict finite implication:
every admissible seed either contradicts a terminal cardinality
bound or, after at most two lifts, exceeds $53/25$.
To see that this persists under a small fixed enlargement, suppose
otherwise. Take a sequence of failing tuples with
$c_{{\rm pr},m}\downarrow c$ and tolerances $\omega_m\downarrow0$.
Write their seed and successor states as
$(q_{0,m},L_{0,m})$, $(q_{1,m},L_{1,m})$, and
$(q_{2,m},L_{2,m})$. A failing tuple obeys, up to $\omega_m$ errors, that
 $a_m\in[-\omega_m,d_{{\rm pr},m}+\omega_m]$, $q_{0,m}=\frac{1+3a_m}{4}+o(1)$, $L_{0,m}=\eta(a_m)+o(1)$, $L_{r+1,m}=M(q_{r,m},L_{r,m},q_{r+1,m})+o(1)$ for $r\in(0,1)$, and 
$L_{r,m}\le T+o(1)$ and $L_{r,m}^2\le K(q_{r,m})+o(1)$ for $r=0,1,2$.
In this, we have set $T=53/25$, and each transition satisfies the perturbed feasibility
constraints. Positive parts define the square roots at perturbed
boundaries. Compactness gives a convergent subsequence, and continuity
then gives a limiting tuple satisfying $a\in[0,d]$ with $(q_0,L_0)=\left(\frac{1+3a}{4},\eta(a)\right)$ and $L_{r+1}=M(q_r,L_r,q_{r+1})$ for $r=0,1$, and $L_r\le T$ and $L_r^2\le K(q_r)$ for $r=0,1,2$.

Its transitions satisfy the unperturbed feasibility constraints.
Thus it has neither a terminal violation nor a threshold above $T$
within two lifts, contradicting Proposition~\ref{prop:finite}(1).

Thus some fixed $c_{\rm pr}>c$ is valid. A numerical value for this
auxiliary buffer is not needed by the rounding algorithm, which never
constructs the pruned graph. The drift statement is unchanged.
Small negative values of $a$ are handled by the same argument and
continuity of the terminal loss at $b=1/4$.

\subsection{Cover masses and uniform errors at depth 404}\label{app:errors}

Pruning is performed once. Its outgoing mass is bounded below by a
positive constant, and its initial two-walk mass is at least
$1/[2(\log\log t)^2]$. We record separately the losses in probability
and in threshold, since they enter different parts of composition.
For projection off a fixed unit direction $h$, write
$w=w^\perp+\ip{w}{h}h$, with $\|w\|\le1$. Then, we have that
$\bigl|\ip{g}{w}-\ip{g}{w^\perp}\bigr|
 \le|\ip{g}{h}|$ and $\Prb[|\ip{g}{h}|>\log t]=2\phib(\log t)
 \le2\exp\!\left(-\frac{(\log t)^2}{2}\right)=:e_t$, with 
$\Prb\!\left[\max_w\ip{g}{w^\perp}\ge s-\log t\right]\ge\Prb\!\left[\max_w\ip{g}{w}\ge s\right]-e_t$. So, for every fixed $A>0$,
\begin{align*}
 \frac{e_t}{t^{-A}}
 &\le2\exp\!\left(-\frac{(\log t)^2}{2}+A\log t\right)
 \longrightarrow0.
\end{align*}
Thus projection preserves, for example, half of any cover mass
$t^{-A}$ once $t$ is sufficiently large. 
At each lift there are $O(\log\Delta)=O(t^2)$ scalar buckets per
parameter, and only constantly many parameters are bucketed. Fix
$B$ so that the product of the relevant bucket counts is at most
$t^B$. Selecting a heavy inner bucket leaves inverse-polynomial mass
for that inner cover; synchronizing their labels among outer indices
costs at most $t^B$ in outer mass. Projection and composition cost
constant factors. Consequently, if $\delta_h$ is the retained outer
cover mass, then for fixed $C_0,B$ and all sufficiently large $t$, we have
$\delta_0\ge t^{-C_0}$,
$\delta_{h+1}\ge\frac{\delta_h}{4t^B}$, $\delta_h\ge4^{-h}t^{-C_0-Bh}
          \ge t^{-C(h+1)}$ for $0\le h\le404$, and $\log\frac1{\delta_h}\le C(h+1)\log t=o(t^2)$ after enlarging the fixed constant $C$. The fresh inner cover at each
endpoint has the constant mass given by pruning, so its probability
bound does not deteriorate through all earlier lifts.

For an inner maximum $F$ with coefficient norms at most one, mean
$\mu$, and cover probability at least $t^{-A}$ at threshold
$\sqrt3t-C_1\log t$, Gaussian upper-tail concentration gives
\begin{align*}
 t^{-A}
 &\le\Prb[F\ge\sqrt3t-C_1\log t]\\
 &\le\exp\!\left[-\frac12
           (\sqrt3t-C_1\log t-\mu)_+^2\right],
           \end{align*} and \begin{align*}
 (\sqrt3t-C_1\log t-\mu)_+
 &\le\sqrt{2A\log t},\\
 \mu&\ge\sqrt3t-C_1\log t-\sqrt{2A\log t}.
\end{align*}
This is the mean estimate used before applying the variance bound.

To control scalar-bucket errors, let $\mathcal U_n$ be any of the
graph-indexed families of unit directions needed at a fixed depth,
with $|\mathcal U_n|\le n^b$ for a fixed $b$. Given $A>0$, choose
$C_A=\sqrt{2(b+A+1)}$. The individual Gaussian tails imply
\begin{align*}
 \Prb\!\left[\max_{u\in\mathcal U_n}|\ip{g}{u}|
                        >C_A\sqrt{\log n}\right]
 &\le\sum_{u\in\mathcal U_n}
          \Prb[|\ip{g}{u}|>C_A\sqrt{\log n}]\\
 &\le2n^b\exp\!\left(-\frac{C_A^2\log n}{2}\right)
  \\
  &=2n^{-A-1}.
\end{align*}
Since $\log\Delta=\Theta(\log n)=\Theta(t^2)$, a coefficient error
bounded by a modulus of continuity $\omega(1/\log\Delta)=o(1)$
contributes, outside this event, at most
\begin{align*}
 \omega(1/\log\Delta)\,
       \max_{u\in\mathcal U_n}|\ip{g}{u}|
 &\le C_A\omega(1/\log\Delta)\sqrt{\log n}
   =o(t).
\end{align*}
All denominators are bounded away from zero on the compact continuing
domain. At a square-root boundary one may use
$|\sqrt{x_+}-\sqrt{y_+}|\le\sqrt{|x-y|}$, so uniform continuity
suffices there as well. Iterating these estimates a fixed number of
times gives normalized errors $e_h(t)$ with
$\max_{h\le404}|e_h(t)|\to0$. In particular, for sufficiently large
$t$, the strict drift from Proposition~\ref{prop:finite}(2) gives
\begin{align*}
 L_{h+1}-L_h
 &>\frac1{1000}-\max_{r\le404}|e_r(t)|\\
 &\ge\frac1{1000}-\frac1{5000}
   =\frac1{1250}.
\end{align*}
This establishes the uniform error convention and the working
increment used in the finite recursion.

\section{Weighted composition and entropy}\label{app:entropy}

This appendix records a weighted form of the composition argument and
makes explicit the distinction between individual branch losses and a
uniform threshold. It is not needed for the numerical recursion.
For a probability vector $\pi$, write
$H(\pi)=\sum_j\pi_j\log(1/\pi_j)$, using natural logarithms.

\begin{lemma}[Weighted composition]\label{lem:weighted}
Let $(x_j)_{j\in J}$ be unit vectors with an $(s,\delta)$-packing
$\mu$ of total mass exactly $\delta>0$, and discard zero-mass indices.
Let $p=\phib(s)$, $\pi_j=\mu(j)/\delta$, and $a_j=\log\frac{p}{\mu(j)}$ with $\overline a=\sum_j\pi_j a_j
             =H(\pi)+\log\frac p\delta$. 
For each $j$, let $F_j(g)$ be the maximum of a fixed finite family
of Gaussian linear forms with coefficient vectors in $x_j^\perp$,
and write $m_j=\E F_j$ and $v_j=\Var(F_j)$.
For $0<\epsilon<1$, with probability at least $(1-\epsilon)\delta$
there is an index $j$ such that $\ip{g}{x_j}\ge s$ and
\[F_j(g)\ge m_j-\sqrt{2v_j\bigl(a_j+\log(1/\epsilon)\bigr)}.\]
If also $m_j\ge m$, $v_j\le v$, and
$0<\eta<1-\epsilon$, then with probability at least
$(1-\eta-\epsilon)\delta$ there is a joint witness satisfying $\ip{g}{x_j}\ge s$ and
\[F_j(g)\ge m-\sqrt{2v\bigl(\overline a/\eta+\log(1/\epsilon)\bigr)}.
\]
\end{lemma}
\begin{proof}
The singleton packing inequalities and the identity
$\mu(j)=\delta\pi_j$ give $0<\mu(j)\le\Prb[\ip{g}{x_j}\ge s]=p$ and $a_j=\log\frac p{\mu(j)}\ge0$. So, we have
\begin{align*}
\overline{a} &=\sum_j\pi_j\log\frac p{\delta\pi_j}\\
 &=\sum_j\pi_j\log\frac1{\pi_j}
   +\left(\sum_j\pi_j\right)\log\frac p\delta
   \\ &=H(\pi)+\log\frac p\delta.
\end{align*}
Define the outer-hit and inner-failure events by $A_j=\{\ip{g}{x_j}\ge s\}$ and $ B_j=\left\{F_j(g)<m_j-\sqrt{2v_j\bigl(a_j+\log(1/\epsilon)\bigr)}\right\}$.
       
For fixed $j$, the coefficient vectors defining $F_j$ lie in
$x_j^\perp$. Thus $F_j$ depends only on the projection of $g$ onto
$x_j^\perp$, which is independent of $\ip{g}{x_j}$. If $v_j>0$,
Lemma~\ref{lem:variance} therefore yields
\begin{align*}
 \Prb(A_j\cap B_j)
 &=\Prb(A_j)\Prb(B_j)\\
 &\le p\exp\!\left[-\frac{2v_j(a_j+\log(1/\epsilon))}{2v_j}\right]\\
 &=\epsilon p e^{-a_j}
   =\epsilon\mu(j).
\end{align*}
If $v_j=0$, then $F_j=m_j$ almost surely, so $B_j$ has probability
zero and the same bound holds. For a one-member inner family, the
tail estimate is the ordinary one-dimensional Gaussian estimate.

The outer packing gives $\Prb(\bigcup_j A_j)\ge\delta$. Moreover,
every outcome in $(\bigcup_j A_j)\setminus\bigcup_j(A_j\cap B_j)$
has an outer witness whose inner inequality succeeds. Consequently
\begin{align*}
 \Prb\!\left[\bigcup_j(A_j\cap B_j^c)\right]
 &\ge\Prb\!\left[\bigcup_j A_j\right]
       -\Prb\!\left[\bigcup_j(A_j\cap B_j)\right]\\
 &\ge\delta-\sum_j\Prb(A_j\cap B_j)\\
 &\ge\delta-\epsilon\sum_j\mu(j)
   =(1-\epsilon)\delta.
\end{align*}

For the uniform threshold, first suppose $\overline a>0$, and retain
$J_\eta=\{j:a_j\le\overline a/\eta\}$. The mass of the discarded
indices is bounded directly by
\begin{align*}
 \sum_{j\notin J_\eta}\mu(j)
 &=\delta\sum_{j\notin J_\eta}\pi_j\\
 &\le\delta\,\frac\eta{\overline a}
                  \sum_{j\notin J_\eta}\pi_j a_j\\
 &\le\delta\,\frac\eta{\overline a}
                  \sum_j\pi_j a_j
   =\eta\delta,
\end{align*} and \begin{align*}
 \Prb\!\left[\bigcup_{j\in J_\eta}A_j\right]
 &\ge\sum_{j\in J_\eta}\mu(j)
   \ge(1-\eta)\delta.
\end{align*}
If $\overline a=0$, positivity of the $\pi_j$ and nonnegativity of
the $a_j$ imply $a_j=0$ for every $j$; take $J_\eta=J$ instead.
In either case, on $B_j^c$ for a retained index,
\begin{align*}
 F_j(g)
 &\ge m_j-\sqrt{2v_j\bigl(a_j+\log(1/\epsilon)\bigr)}\\
 &\ge m-\sqrt{2v\bigl(\overline a/\eta+\log(1/\epsilon)\bigr)}.
\end{align*}
Subtracting the failures over the retained set gives
\begin{align*}
 \Prb\!\left[\bigcup_{j\in J_\eta}(A_j\cap B_j^c)\right]
 &\ge\Prb\!\left[\bigcup_{j\in J_\eta}A_j\right]
       -\sum_{j\in J_\eta}\Prb(A_j\cap B_j)\\
 &\ge(1-\eta)\delta
       -\epsilon\sum_{j\in J_\eta}\mu(j)\\
 &\ge(1-\eta-\epsilon)\delta,
\end{align*}
which is the second assertion.
\end{proof}

In a uniform packing, $a_j=\log(|J|p/\delta)$ is constant, so no
truncation is needed. An arbitrary cover need not admit a uniform
packing. Nevertheless, for an outer cover of probability $\delta$,
putting $B=\log(2|J|p/\delta)$ gives
\begin{align*}
 \Prb\!\left[\bigcup_j
      \left(A_j\cap\{F_j<m_j-\sqrt{2v_jB}\}\right)\right]
 &\le\sum_j p e^{-B}\\
 &=|J|p\,\frac\delta{2|J|p}
   =\frac\delta2.
\end{align*}
This leaves probability at least $\delta/2$ for a joint witness, as
in Lemma~\ref{lem:compose}. Thus that estimate requires only
cardinality, while the weighted form keeps track of individual masses.
The identity for $\overline a$ is an entropy interpretation of the
average logarithmic loss. A bound on Shannon entropy alone does not
control each branch loss; obtaining the displayed uniform bound costs
the stated truncation factor $1/\eta$.

For comparison, the Gaussian isoperimetric route starts with an inner
$(s_2,\delta_2)$-cover of vectors of norm at most $r>0$. Let
$F(g)=\max_{u\in U}\ip{g}{u}$, let
$A=\{g:F(g)\ge s_2\}$, and let $A^{(a)}$ be its Euclidean
$a$-enlargement. If $x\in A$ and $\|y-x\|\le a$, then
\begin{align*}
 F(y)
 &=\max_{u\in U}\bigl(\ip{x}{u}+\ip{y-x}{u}\bigr)\\
 &\ge F(x)-r\|y-x\|
   \ge s_2-ra,
\end{align*}
and $A^{(a)}\subseteq\{g:F(g)\ge s_2-ra\}$. Writing $\gamma$ for standard Gaussian measure, Borell's
inequality~\cite{Borell} gives that for all $a\geq 0$,
\begin{align*}
 \Prb[F(g)\ge s_2-ra]
 &\ge\gamma(A^{(a)})\\
 &\ge\Phi\bigl(\Phi^{-1}(\gamma(A))+a\bigr)\\
 &\ge\Phi\bigl(\Phi^{-1}(\delta_2)+a\bigr),
\end{align*}
where $\Phi=1-\phib$ and $0<\delta_2<1$.
This is the enlargement step in the cover-composition argument
of Bansal, Huang, and Lee~\cite[Lemma~4.14]{BHL}.
The present proof uses Chen's variance-scale lower tail together
with per-index independence and a union bound. The comparison
after Lemma~\ref{lem:compose} explains its relation to the Ehrhard-based
composition bound of Narang and Tang~\cite{narang2026improvedsdpcoloring3colorable}.

\end{document}